\documentclass[12pt]{article}

\usepackage{graphicx}
\usepackage{a4wide}

\usepackage[linesnumbered, longend ]{algorithm2e}
\usepackage{amssymb}
\usepackage{amsmath}

\usepackage{amsthm}
\usepackage[english]{babel}
\usepackage{enumitem}
\usepackage{float}
\usepackage{epstopdf}
\usepackage{tikzpagenodes}
\usepackage{pgf, tikz}
\usetikzlibrary{arrows, automata, backgrounds, patterns, calc, shapes, spy, external}

\usepackage{subcaption}
\usepackage{slashbox}
\usepackage{url}
\usepackage{xcolor}
\usepackage{marginnote}
\usepackage{todonotes}
\usepackage[colorlinks, linkcolor=red!50!black, citecolor=green!50!black, urlcolor=blue!50!black]{hyperref}
\usepackage{theoremref}
\usepackage[utf8]{inputenc}
\usepackage[numbers, square, comma]{natbib}

\newcommand{\R}{\mathbb R}

\renewcommand{\S}{\mathbb S}

\renewcommand{\phi}{\varphi}
\newcommand{\iks}{\mathfrak{X}}
\renewcommand{\div}{\operatorname{div}}
\renewcommand{\d}{\mathrm{\,d}}
\newcommand{\e}{\mathrm{e}}
\newcommand{\iid}{\stackrel{\mathrm{i. i. d.}}{\sim}}
\newcommand{\Beta}{\mathrm{B}}
\newcommand{\Eta}{\mathrm{H}}

\newcommand{\dcup}{\,\dot\cup\,}
\newcommand{\ind}{\mathbf{1}}

\newcommand{\sub}{\subseteq}

\newcommand{\dmin}{d_{\min}}

\xdefinecolor{dgreen}{rgb}{0,0.5,0}
\xdefinecolor{dblue}{rgb}{0,0,0.5}
\xdefinecolor{dred}{rgb}{0.5,0,0}

\DeclareMathOperator{\E}{\mathbb{E}}

\newtheorem{theorem}{Theorem}[section]

\theoremstyle{definition}
\newtheorem{definition}[theorem]{Definition}
\newtheorem{algo}[theorem]{Algorithm}
\newtheorem{remark}[theorem]{Remark}
\newtheorem{example}[theorem]{Example}

\begin{document}
	
	\title{Characteristic and Necessary Minutiae in Fingerprints} 
	
	\author{Johannes Wieditz\thanks{Institute for Mathematical Stochastics, Georg-August-Universtität Göttingen, Goldschmidtstraße 7, 37077 Göttingen, Germany, johannes.wieditz@uni-goettingen.de} \and Yvo Pokern\thanks{Department of Statistical Science, University College London, Gower Street, London WC1E 6BT, United Kingdom, y.pokern@ucl.ac.uk} \and 
		Dominic Schuhmacher\footnotemark[1] \and Stephan Huckemann\footnotemark[1]}
	
	\maketitle
	
	\begin{abstract}
		\textbf{Abstract } Fingerprints feature a ridge pattern with moderately varying ridge frequency (RF), following an orientation field (OF), which usually features some singularities. 
		Additionally at some points, called minutiae, ridge lines end or fork and this point pattern is usually used for fingerprint identification and authentication. 
		Whenever the OF features divergent ridge lines (e.g.\ near singularities), a nearly constant RF necessitates the generation of more ridge lines, originating at minutiae.
		We call these the necessary minutiae. It turns out that fingerprints feature additional minutiae which occur at rather arbitrary locations. We call these the random minutiae or, since they may convey fingerprint individuality beyond the OF, the characteristic minutiae. In consequence, the minutiae point pattern is assumed to be a realization of the superposition of two stochastic point processes: a Strauss point process (whose activity function is given by the divergence field) with an additional hard core,  and a homogeneous Poisson point process, modelling the necessary and the characteristic minutiae, respectively. We perform Bayesian inference using an MCMC-based minutiae separating algorithm (MiSeal). In simulations, it provides good mixing and good estimation of underlying parameters. In application to fingerprints, we can separate the two minutiae patterns and verify by example of two different prints with similar OF that characteristic minutiae convey fingerprint individuality.
		
	\end{abstract}%
	
	\section{Introduction}
	\label{sec:intro}
	
	Authentication and identification by fingerprints is increasingly popular in a wide variety of applications, for example in smart phones and internet banking on the commercial side and in border control on the governmental side. Also on the side of forensics, fingerprint analysis is enjoying undiminished attention. For an overview, see \cite{maltoni_handbook_2009}. 
	
	Fingerprints feature a ridge line pattern inducing an undirected 
	\emph{orientation field} (OF), with zero to four singularities (zeroes or 
	poles of an underlying otherwise smooth and non-vanishing orientation 
	field, cf. \cite{huckemann_global_2008}, resulting in nearby ridges having 
	high curvature) called \emph{cores} (where neighbouring ridge lines go 
	around an ending ridge line) and \emph{deltas} (where three ridge lines 
	meet). The \emph{ridge frequency} (RF) varies moderately over the print and 
	changes rapidly only near singularities. Points where ridge lines fork or 
	end are called \emph{minutiae}. Figure~\ref{figure:minutia-types} provides 
	an example of a fingerprint featuring a double core called a \emph{whorl} 
	(around which ridge lines circle) and a delta near the right bottom corner. 
	
	\begin{figure}[b!]
		\centering
		\includegraphics[width = .9\textwidth]{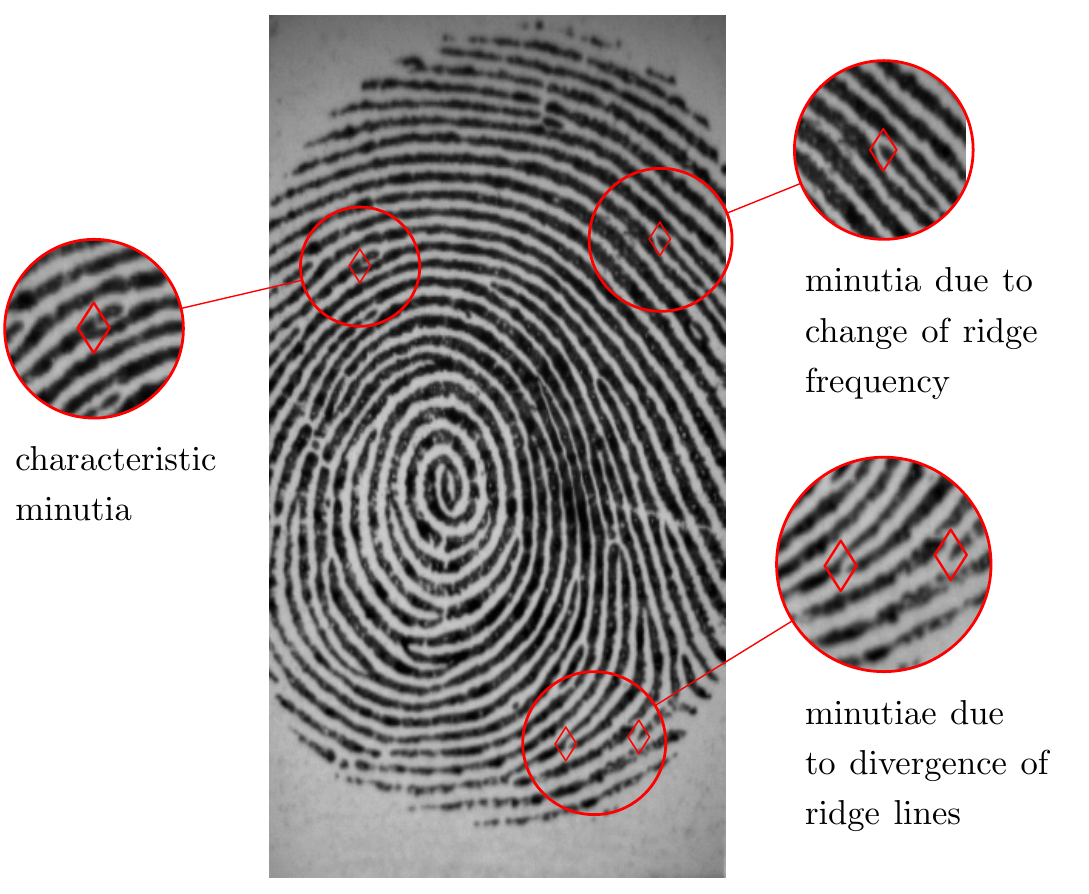}
		\caption{\it Different ``reasons'' for minutiae ($\Diamond$) in imprint 6 of finger 7 of DB2 in FVC2002 from \cite{maio_fvc2002_2002}.}
		\label{figure:minutia-types}
	\end{figure}
	
	To compare two fingerprints, usually each is reduced to its minutiae point pattern, often augmented by OF information (e.g.\ marking local orientation of the field). 
	Typically, challenges arise due to partial imprints (as in Figure~\ref{figure:minutia-types}, where, due to the global structure of fingerprint OFs, there is another delta further out on the lower left side, not observed in the print), low image quality and variable distortions arising from taking a planar image of a two-dimensional curved surface. In our work we demonstrate that a better understanding of the minutiae pattern has the potential of developing improved matching algorithms. 
	
	Remarkably, the details of the biological fingerprint formation process which takes place during early gestational weeks, are still largely in the dark. Modelling the formation of the OF with its singularities and of the RF, varying only within a small interval, by  expanding patterns satisfying suitable partial differential equations, \cite{KuckenNewell2007} observed in simulations that minutiae \emph{mostly occur in two circumstances}: when ridge lines diverge with new ridges inserted and when almost parallel ridge lines meet. They further observed that minutiae positions were quite sensitive to initial conditions and \cite{KuckenChampod2012} explained this biologically by \emph{small differences in the Merkel cell distribution}. 	These (random) differences can be very subtle as \cite{newman_finger_1930} 
	noted much earlier: 
	{\it There are, however, numerous instances in which the prints of two of more homologous fingers are so nearly identical as to be indistinguishable to the naked eye. [...]
		it is possible only by using considerable magnification to discover differences in the branching of ridges and breaks in ridge continuity. Differences of this sort, however, are certain to be found, and afford an easy means of identification.}
	
	Inspired by this, we focus theoretically and empirically on the interaction of minutiae, OF and RF and argue that OF divergence and RF changes geometrically necessitate minutiae which we hence call \textit{necessary} minutiae. Statistical analysis endorses the above mentioned observations of additional \emph{random} minutiae which are independent of the underlying smoothed OF and RF, cf.\ Figure \ref{figure:twin-fingerprints} which shows prints of monozygotic twins from~\cite{newman_finger_1930}.

	\begin{figure}[b!]
		\centering
		\begin{subfigure}{.48\linewidth}
			\centering
			\includegraphics[width=\linewidth]{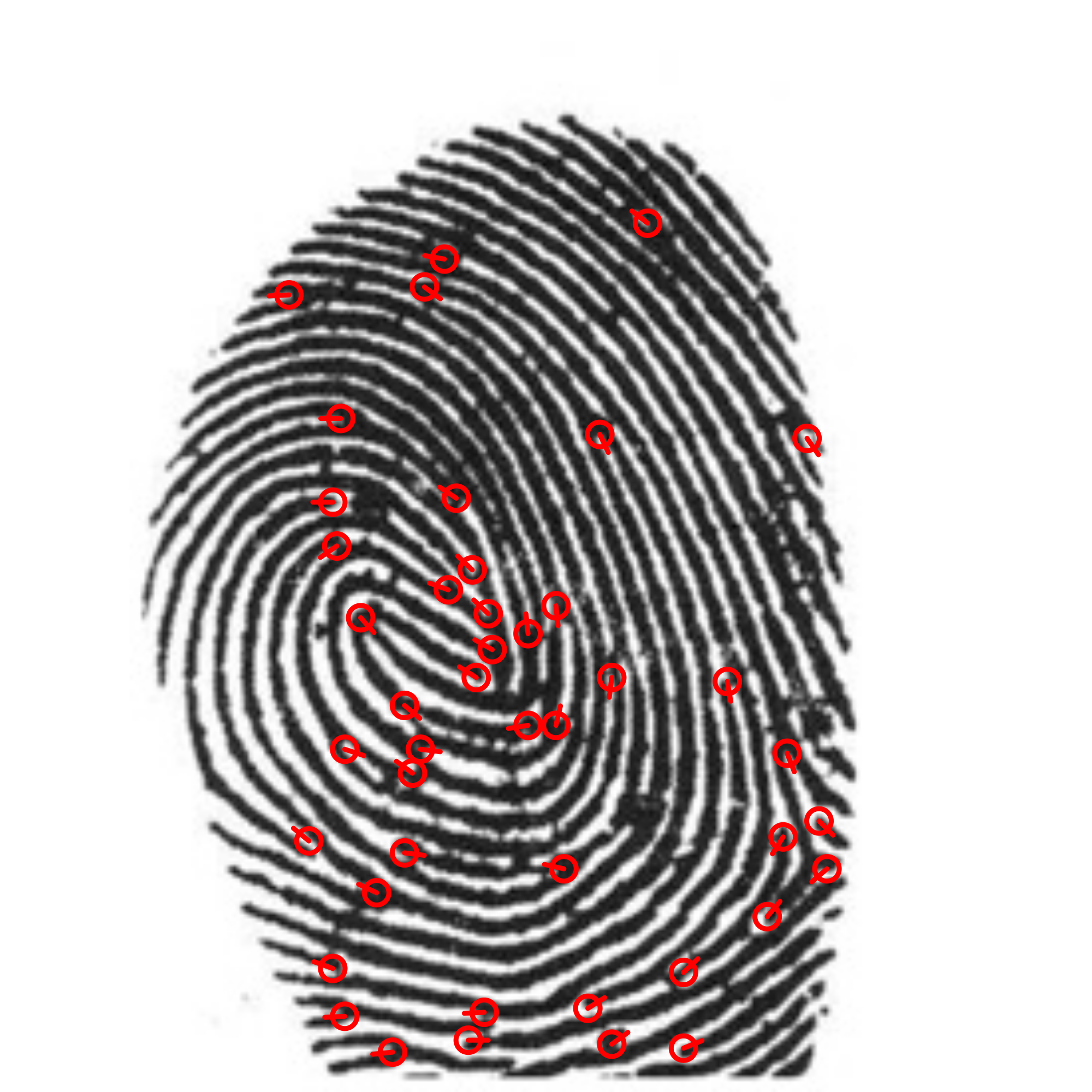}
			\caption{}
			\label{figure:twin2a}			
		\end{subfigure}%
		\hfill
		\begin{subfigure}{.48\linewidth}
			\centering
			\includegraphics[width=\linewidth]{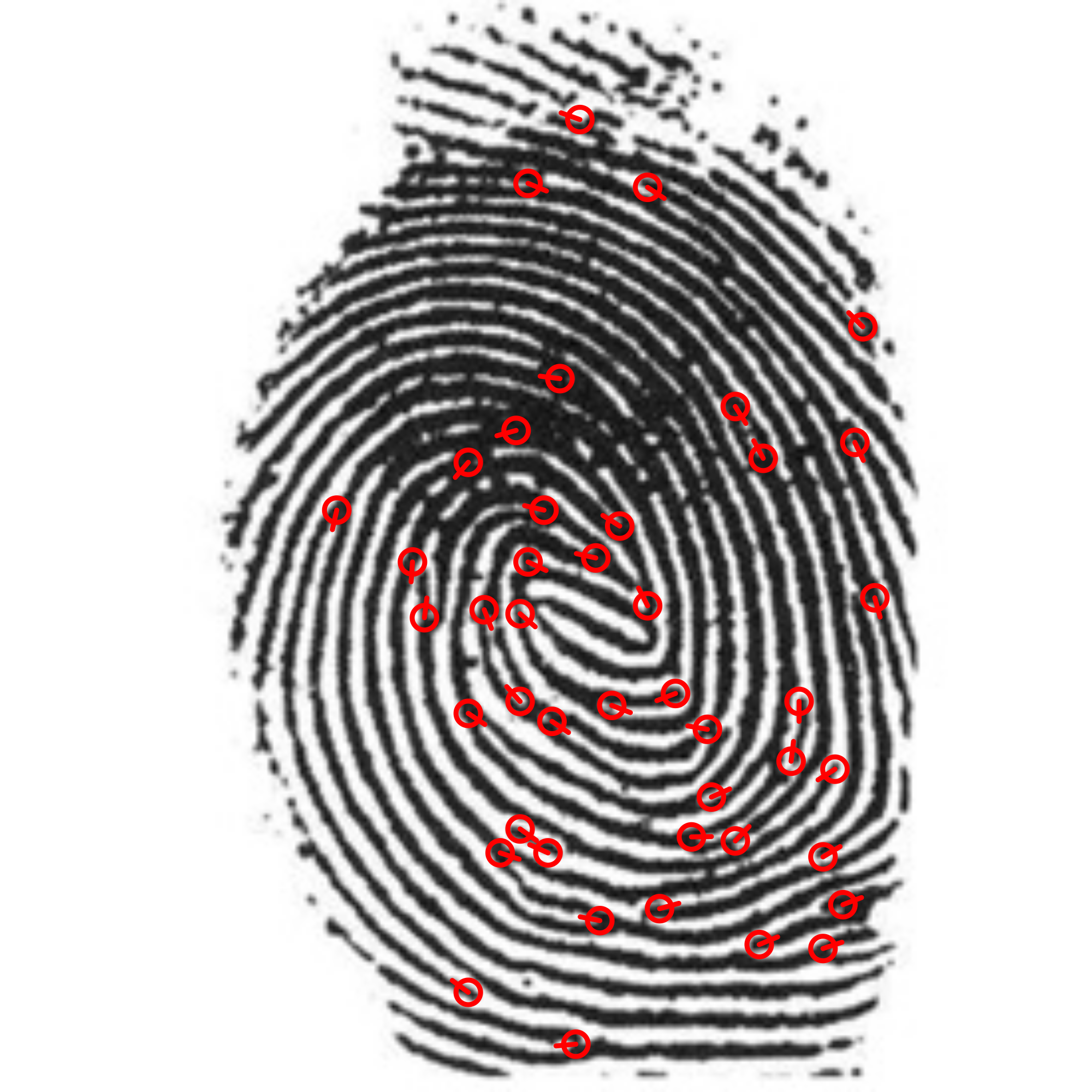}
			\caption{}
			\label{figure:twin2b}
		\end{subfigure}\textbf{}
		\caption{\textit{Fingerprints of each of one twin of a pair of monozygotic twins from~\cite{newman_finger_1930}, labelled 14a and 14b. Even though the fingerprints between the twins resemble in OF and RF, they do exhibit different minutiae (here: manually marked) e.\ g.\ in the region left of the delta in the lower right corner.}}
		\label{figure:twin-fingerprints}
	\end{figure}
	
	For imprints with similar OF and RF leading to similar necessary minutiae point patterns, these random minutiae may carry individuality information and for this reason, we also call them \emph{characteristic} minutiae. 
	
	In the following, we lay the mathematical foundations and provide an algorithm separating the superposition of an inhomogeneous Strauss point process (modelling the necessary minutiae) and a homogeneous Poisson process (modelling the random minutiae) which is an improved and extended version of \cite{redenbach_classification_2015,rajala_variational_2016}, which separated a homogeneous Strauss point process from a homogeneous Poisson point process using \emph{Markov-Chain-Monte-Carlo} (MCMC) and variational Bayes methods.
	In application to manually re-marked fingerprints our minutiae separating algorithm (MiSeal) finds the presence of random minutiae. Furthermore, in an exemplary analysis of two different imprints with similar OF, we find that these minutiae are indeed characteristic: excluding them results in more similar minutiae patterns than excluding the same number of minutiae at random.
	
	Application of our methods to a large fingerprint database and thus to a use case (e.g. improved identification in forensics) is beyond the scope of this work because automatic minutiae extraction is often imprecise in practice. Thus, our algorithms would require tuning and additional adaptation to such applications which is the subject of current research and left for future publications.
	
	OF, RF and minutiae extraction as well as matching procedures are well kept proprietary secrets of commercials firms. Instead of relying on such closed-source codes, we  
	binarize fingerprint images using the algorithm from \cite{thai_filter_2016}, estimate OF and RF using our own implementation in Java based on~\cite{hong_fingerprint_1998}  and manually extract minutiae.
	
	The extracted point patterns have been compared using the \emph{minutiae cylinder code} (MCC) from  \cite{cappelli_mcc_2010}. Our software can be found at~\cite{wieditz_fingerprint_2020}.
	
	To the best of our knowledge, including divergence information in minutiae matching has not gained attention in the literature so far. Conversely, the observation that minutiae cause high local divergence of the ridge flow field has been used in \cite{nikodemusz-szekely_image_1993} to locate them.
	
	In the following Section~\ref{sec:mdf} we formalize the concepts of OF and RF divergence for fingerprints. In Section~\ref{sec:exofrandom} we provide evidence for the existence of random minutiae. In Section~\ref{sec:model} we specify our point process models and in Section~\ref{sec:mcmc} we develop the MCMC-based algorithm MiSeal for separation and estimation of the model's parameters. With the help of simulated data based on divergence maps estimated from 20 high quality fingerprints of the database FVC2002 DB1, cf.\ \cite{maio_fvc2002_2002}, MiSeal is evaluated in Section~\ref{sec:test} and we show that the database's fingerprints indeed feature necessary and random minutiae. Section~\ref{sec:proof} gives an example of two similar fingerprints from the two monozygotic twins from Figure~\ref{figure:twin-fingerprints} where the random minutiae carry characteristic information for distinguishing the prints.

	\section{A Formula for Necessary Minutiae}
	\label{sec:mdf}
	
	We model the OF induced by an observed ridge line pattern as detailed in~\cite{huckemann_global_2008} by a unit length orientation field $O: \iks \to \mathbb{R}P^1 \cong \S^1$ (angles between $0$ and $\pi$ in the real projective space of dimension one -- which is topologically a circle) that is well-defined, non-vanishing and $\mathcal{C}^2$ apart from isolated singularities: zeros of $O$ result in deltas in the ridge pattern and poles of $O$ result in cores and whorls, cf. \cite{huckemann_global_2008}. Here, $\iks\sub \R^2$ denotes the \textit{region of interest}, i.e.\ the part of the image containing the fingerprint, which is assumed to be compact with piecewise smooth boundary.\footnote{Here and later on, we mean by this that the boundary is the image of a simple closed curve that is piecewise $\mathcal{C}^2$.} Every orientation in $O$ has two well defined directions, namely the original orientation and the original orientation plus $\pi$. In simply connected regions $A\subseteq \iks$ not containing any of the singularities, we can pick a continuous selection of directions from $O$ which we call $\vec F = \vec F_A$. See~\cite{huckemann_global_2008} for more details. 
	
	Furthermore, the ridge pattern features a locally varying ridge frequency (RF), which we model as a $\mathcal{C}^2$-function $\Phi: \iks \to (0,\infty)$. Algorithmically, $\Phi(z)$ is obtained by centring a line segment at $z$ that is aligned orthogonally to $O$ and dividing the number of ridges the line segment crosses by its length; subsequent smoothing results in a $\mathcal{C}^2$-function. Empirically, $\Phi$ varies only within a small interval (its inverse, the inter-ridge distance, is between 6 and 15 pixels in all of the data considered).
	
	Ideally, the integral 
	$$ \int_\gamma \Phi(z) \d z$$
	along a curve $\gamma$ orthogonal to the field $O$ closely approximates the number of ridges crossing that curve. This motivates the following general definition which is further illustrated in Examples \ref{ex:1} and \ref{ex:2} below; see also Figure \ref{figure:example-annular-sector}.
	
	\begin{definition}
		For simply connected compact $A\subseteq \iks$ that does not contain any of the singularities of $\vec F$ and has piecewise smooth boundary $\partial A$ with (piecewise well-defined) outwards pointing normal $\vec n:\partial A \to \S^1$, we call		
		\begin{align}
		\label{def:1}
		m(A) := \left| \int_{\partial A} \Phi(z)\left\langle \vec F(z), \vec{n}(z) \right\rangle \d z \right|
		\end{align}
		the (usually non-integer-valued) \textit{number of geometrically 
			necessary minutiae} in $A$, for short the \textit{necessary minutiae 
			number}.
	\end{definition}
	
	As $A$ contains no singularities, ridges near $A$ carry a common directional flow induced by $\vec F$. Then the  necessary minutiae number $m(A)$ counts the absolute difference of numbers of ridges entering $A$ and leaving $A$, each weighted by the cosine of the angle between ridge and outwards pointing normal (counted fully if they intersect the boundary of $A$ perpendicularly). Taking the absolute value of the difference provides independence of the particular flow direction $\vec F$ of $O$ chosen. Thus $m(A)$ counts the number of minutiae \emph{necessary} due to the geometry of the OF and the RF, see also Figure~\ref{figure:example-annular-sector}. Minutiae in $A$ annihilating each other, e.g.\ due to a ridge beginning and ending in $A$, are not counted.

	Applying the divergence theorem, see \cite[Theorem~XII.3.15 and Remark~XII.3.16(c)]{amann3engl} or \cite[Theorem~16.7]{adams_calculus_2016} for a more direct formulation, we obtain	
	\begin{align*}
	\notag \int_{\partial A} \left\langle \Phi(z) \vec F(z), \vec{n}(z) \right\rangle \d z &= \int_A \div\left( \Phi \vec F \right) (z) \d z \\
	&= \int_A \Phi(z) \div \vec F(z) \d z + \int_A \left\langle \nabla \Phi(z), \vec F(z) \right\rangle \d z ,
	\end{align*}
	yielding
	\begin{align}
	\label{eq:thm3}
	m(A) = \left| \int_{A} \Phi(z) \div \vec F(z) \d z + \int_A \left \langle \nabla \Phi(z),  \vec F(z) \right \rangle \d z \right|.
	\end{align}	
	The first term, 
	$$ \int_{A} \Phi(z) \div \vec F(z) \d z, $$ 
	captures the effect of the \textit{OF divergence}, whereas the second term,
	$$\int_A \left \langle \nabla \Phi(z),  \vec F(z) \right \rangle \d z,  $$
	captures the \textit{RF divergence}. 
	
	It may happen that, following the field in one direction, the inter-ridge distances decrease as the field lines converge (e.g.\ lines and the spaces between them get thinner). Then, RF divergence and OF divergence have different signs, nearly cancelling each other, yielding $m(A) \approx 0$. The minutiae number is always non-negative due to the absolute values taken in (\ref{def:1}) and (\ref{eq:thm3}), in particular making the sum of divergences in (\ref{eq:thm3}) independent of the specific direction chosen.
	
	Recall that a set $A \subseteq \mathbb{R}^2$ is called \emph{star-shaped} with respect to $z_0 \in \mathbb{R}^2$ if $t z_0 + (1-t) z \in A$ for all $z \in A$ and $t \in [0,1]$, cf. \cite[p.~314~ff.]{amann2engl}. We write $r(A) := r_{z_0}(A) := \sup_{z \in A} \|z-z_0\|$ for the radius of such a set and $|A|$ for its area (if it is measurable).
	
	\begin{theorem}\label{thm:mf}
		Let $z_0 \in \iks$ be fixed.
		Suppose that $A\subseteq \iks$ is a compact set that is star-shaped w.r.t.\ $z_0$, does not contain any of the singularities of $\vec F$ and has piecewise smooth boundary~$\partial A$. Then		
		\begin{align*}
		m(A) =  \left|\Phi(z_0) \div \vec F(z_0) + \left\langle \nabla \Phi(z_0), \vec F(z_0) \right\rangle\right| \cdot |A| + o(|A|) \quad \text{as $r(A) \to 0$.}
		\end{align*}
	\end{theorem}
	
	\begin{proof}
		Since $A$ is simply connected, compact and does not contain any of the singularities of $\vec F$, the function $f(z) = \div(\Phi \vec F)(z)$ is $\mathcal{C}^1$ with bounded derivative where the derivative is to be suitable interpreted at boundary points of $\iks$. Using the reverse triangle inequality and the multivariate mean value theorem we have for every $z \in A$, 
		\begin{align*}
		&\left| m(A) - \left|\Phi(z_0) \div \vec F(z_0) + \left\langle \nabla \Phi(z_0), \vec F(z_0) \right\rangle\right| \cdot |A| \right| =
		\left| \left|\int_A f(z) \d z\right| - |f(z_0)||A| \right|\\ & \le \left| \int_A f(z) \d z - f(z_0)|A| \right| =\left| \int_A \nabla f(\xi_z)^\top(z-z_0) \d z \right| \\
		&\le \sup_{\xi \in A} \left\| \nabla f(\xi) \right\| \, \sup_{z\in A} \left\| z-z_0 \right\| \, |A|.
		\end{align*}
		Now, dividing by $|A|$ and letting $r(A) \to 0$ yields the assertion. Above, $\xi_z$ is a measurable selection from the measurable set $\{\xi \in A: 0 = f(z) - f(z_0) -   \nabla f(\xi)^\top(z-z_0)\}$, for instance one with minimal first and, if necessary, also with minimal second component.	
	\end{proof}
	
	This theorem motivates the definition of the necessary minutiae intensity governed by the sum of local OF divergence and local RF divergence.
	
	\begin{definition}
		For $z_0 \in \iks$ outside the set of singularities of $\vec F$, call		
		\begin{align}
		\label{eq:min.estimate}
		\mu(z_0) =  \left|\Phi(z_0) \div \vec F(z_0) + \left\langle \nabla \Phi(z_0), \vec F(z_0) \right\rangle\right| 
		\end{align}
		the \textit{intensity of necessary minutiae} at $z_0$.
	\end{definition} 
	
	\begin{figure}[t!]
		\centering
		\begin{subfigure}[t]{0.5\textwidth}
			\centering
			\includegraphics{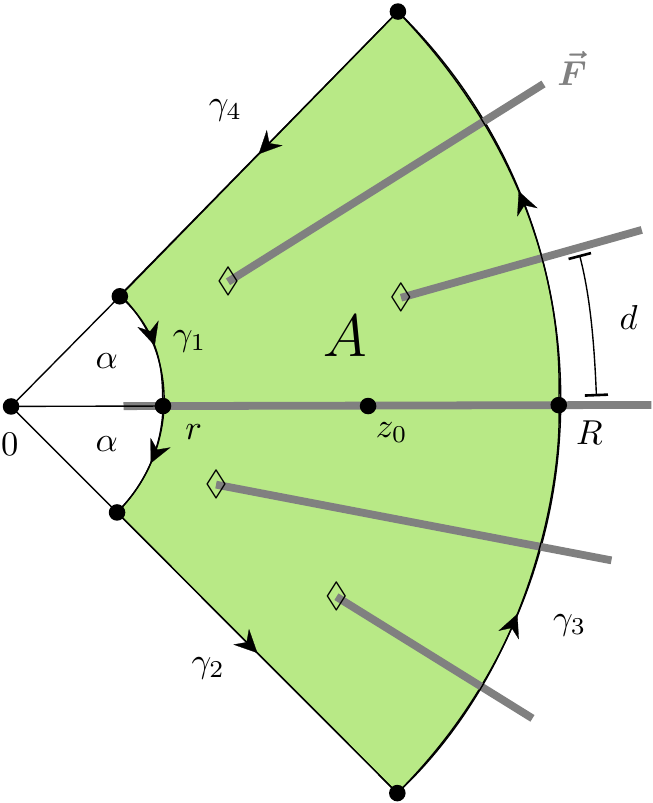}
			\caption{}
			\label{figure:example-annular-sector-a}
		\end{subfigure}%
		~ 
		\begin{subfigure}[t]{0.5\textwidth}
			\centering
			\includegraphics{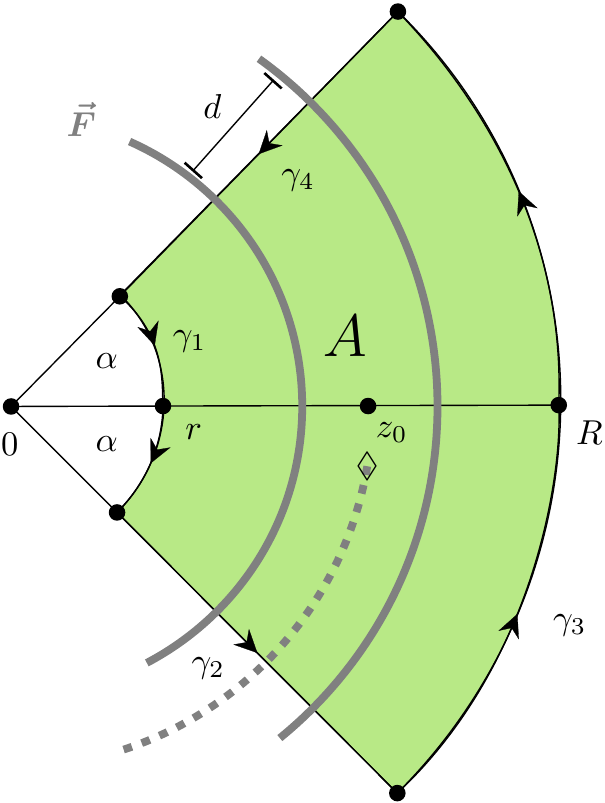}
			\caption{}
			\label{figure:example-annular-sector-b}
		\end{subfigure}
		\caption{\it Ridge pattern within an annular sector $A$ (green) around $z_0$ (on the first axis between $r$ and $R$)  generated by the OF $\vec F: z=(x,y) \mapsto \frac{(x,y)}{\|z\|}$ (left) and the field $\vec F: z=(x,y) \mapsto \frac{(y,-x)}{\|z\|}$ (right), which is orthogonal to the field on the left. The number of minutiae in $A$ is given by the number of new ridges emerging in~$A$~($\Diamond$).}
		\label{figure:example-annular-sector}
	\end{figure}
	
	\begin{example}
		\label{ex:1}
		Fix $\varepsilon >0$, consider $\iks = \{z=(x,y)\in \R^2 : \|z\| \geq \varepsilon\}$ and $\vec F: \iks \to \S^1, z\mapsto \frac{z}{\|z\|}$ pointing radially away from the origin. For $z_0\in \iks$ and $\varepsilon<r< R$ consider an annular sector		
		\begin{align}
		\label{eq:6}
		A := \left\{ z \in \R^2 \,: \, \big\vert \, \angle(z, z_0)\,\big\vert \le \alpha, r \leq \|z\| \leq R \right\}
		\end{align}
		of opening angle $\alpha \in [0,\frac\pi2]$. Figure~\ref{figure:example-annular-sector} shows the situation for a nearly constant RF $\Phi(z)\approx \frac{1}{d}$ where one ridge line enters from the left and five ridge lines leave on the right, giving rise to four minutiae (marked with $\Diamond$) inside $A$. Along the circular arcs $\gamma_1$ (of length $2\alpha r$) and $\gamma_3$ (of length $2\alpha R$) the outwards pointing normal $\vec n$ of $\partial A$ is first antiparallel and then parallel to the field, while on the radial arcs $\gamma_2$ and $\gamma_4$ the outwards pointing normal $\vec n$ of $\partial A$ is orthogonal to the field. With $\Phi(z) = \frac{1}{d}$, this gives
		$$ m(A) =  2\alpha\, \frac{R-r}{d}\,.$$
		Indeed, with $\partial_x \frac{x}{\sqrt{x^2+y^2}} = \frac{y^2}{(x^2+y^2)^{3/2}}$ and $\partial_x \frac{y}{\sqrt{x^2+y^2}} = \frac{x^2}{(x^2+y^2)^{3/2}}$ we have simply $ \div \vec F(z) = \frac{1}{\|z\|}$, and in the presence of OF divergence only, introducing polar coordinates,
		$$ \int_A \Phi(z) \,\div \vec F(z)\,\d z = \frac{2\alpha}{d} \int_r^R\frac{\rho \d \rho}{ \rho} = 2\alpha\, \frac{R-r}{d}\,.$$
		
		Supposing that there were fewer than four minutiae observed in Figure \ref{figure:example-annular-sector-a}, then fewer than five ridge lines would cross $\gamma_3$. This would necessitate a lower ridge frequency on $\gamma_3$ than on $\gamma_1$, yielding $\langle \nabla \Phi(z),\vec F(z)\rangle < 0 < \div \vec F(z)$, so that the OF divergence would be cancelled partially (or in total) by the RF divergence.
	\end{example}
	
	\begin{example}
		\label{ex:2}
		With $\iks$ and $A$ from Example \ref{ex:1}, consider now the field $\vec F: \iks \to \S^1, z=(x,y) \mapsto \frac{(y,-x)}{\|z\|}$, which is perpendicular to the field from Example \ref{ex:1}, cf. Figure~\ref{figure:example-annular-sector-b}. Since $\div \vec F(z) = 0$, this field is divergence free and for constant ridge frequency $\Phi(z) = \frac{1}{d}$ we do not observe any minutiae in $A$, i.e.	
		$$m(A) = 0\,.$$
		Indeed, now the field is orthogonal to the outwards pointing normals of $\partial A$ along the circular arcs $\gamma_1$ and $\gamma_3$ while it is parallel and antiparallel, respectively, on the radial arcs $\gamma_2$ and $\gamma_4$, which are of equal lengths, so their contributions to $\Phi(z) \langle \vec F(z),\vec n(z)\rangle$ cancel. 
		
		If a minutia was observed within $A$, then it would be due to the RF divergence of the non-constant RF, namely of RF higher on $\gamma_2$ than on $\gamma_4$, as depicted with the dotted ridge in Figure \ref{figure:example-annular-sector-b}. 
	\end{example}
	
	\section{The Existence of Random Minutiae}
	\label{sec:exofrandom}
	
	Having found a formula predicting the number of necessary minutiae given the OF's and the RF's divergence, we investigate in this section whether there are additional minutiae in fingerprint patterns not explained by OF and RF divergence.
	
	To this end, we preprocess 20 high quality fingerprints from the database FVC2002~DB1\footnote{fingers labelled
	1\_1,
	2\_8,
	7\_1,
	9\_8,
	13\_4,
	22\_4,
	25\_2,
	26\_2,
	28\_4,
	31\_5,
	34\_1,
	35\_6,
	53\_6,
	57\_3,
	59\_2,
	65\_4,
	66\_2,
	76\_6,
	89\_5 and
	100\_6}
	using the algorithm from \cite{thai_filter_2016} in order to obtain 
	enhanced and binarized versions of the images and the regions of interest. 
	We then manually mark the minutiae, subdivide each image into approx. $100$ 
	rectangular patches $A$ (aspect ratio taken from the images), cf. Figure 
	\ref{figure:Intensity-Twins} (there the patches have been chosen a little 
	smaller), and count the number of minutiae in these patches.
	
	For comparison, we compute the intensity of necessary minutiae based on the approximation in Theorem~\ref{thm:mf} using $\Phi, \nabla \Phi, \vec F, \div \vec F$ which, in turn, are obtained by smoothing with a Gaussian kernel. Patches too close to a singularity are discarded, because due to large derivatives of these quantities the approximations are typically bad; see the proof of Theorem~\ref{thm:mf}.  
	
	The black dots in Figure~\ref{figure:PoissonRegression} display the actual minutiae counts against the necessary minutiae numbers $m(A)$ over the different patches $A$ in all 20 imprints, i.e.\  $m(A)$ is the (not necessarily integer valued) number of minutiae we expect in $A$ if there are only necessary minutiae. We perform a Poisson regression with 
	identity link, i.e.\ we model the expectation $\mu(A)$ of the actual minutiae count in $A$ 
	as $\mu(A) = \beta_0 + \beta_1 m(A)$ and determine $\hat \beta_0, \hat 
	\beta_1$ by maximum likelihood estimation. As a word of caution we point out that the difference of the data to its regression line is indeed consistent with a Poisson regression for count data, see e.g.\ \cite{cameron_regression_2013}; we do not 
	perform ordinary least squares regression here. 
	
	\begin{figure}[t!]
		\centering
		\includegraphics[width =\linewidth]{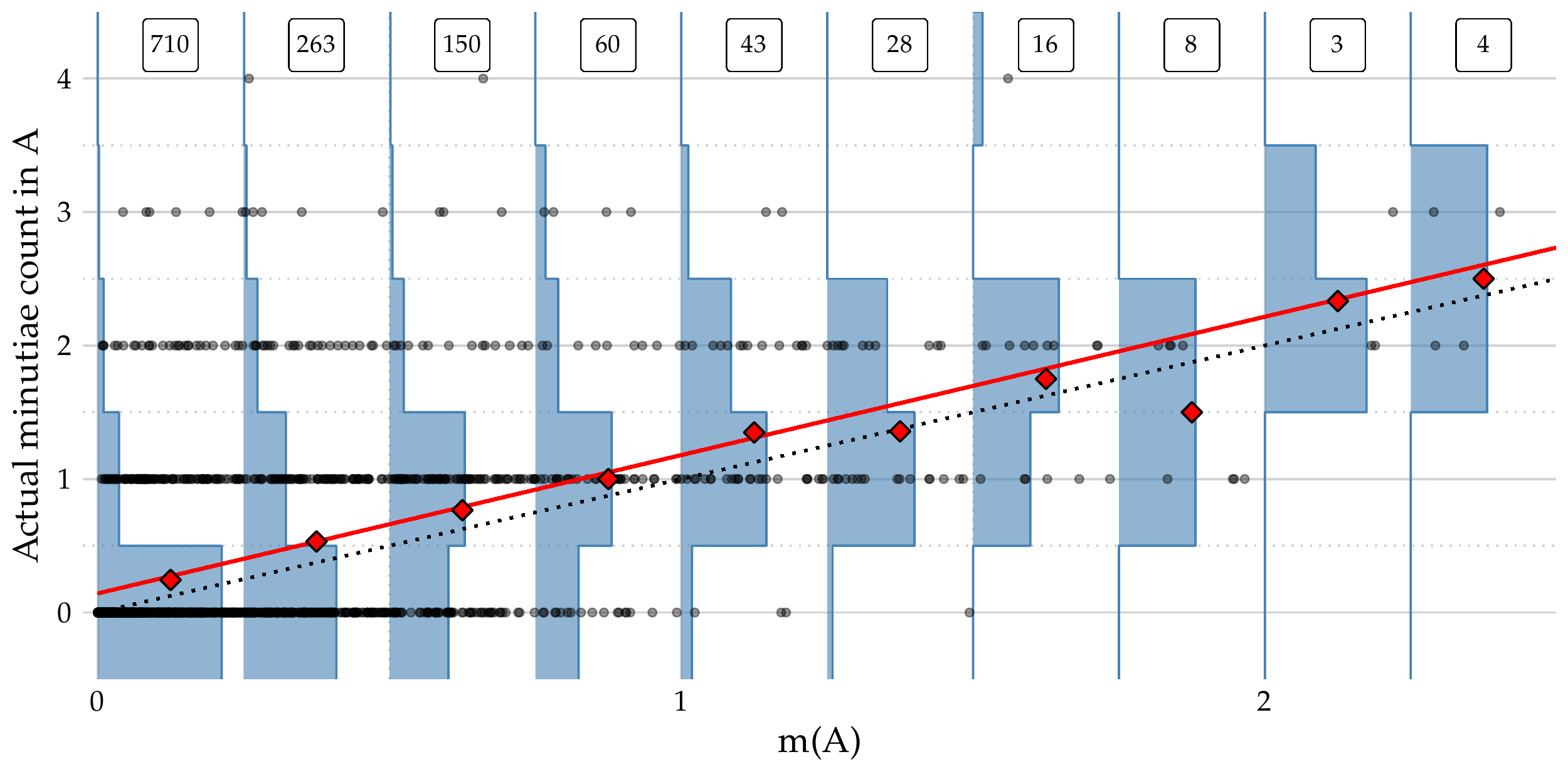}
		\caption{\textit{Poisson regression (red line) for actual minutiae count (black dots) in patches $A$ (of approx 1451 pixels, cf. Figure \ref{figure:Intensity-Twins} and Remark~\ref{rem:lambda0}) versus necessary minutiae number $m(A)$ computed by Formula~\eqref{eq:thm3}. The dotted black line (identity function) shows the relation we would expect if there were no random minutiae. The probability mass function of the counts within bins of width 0.25 is depicted in blue; their mean is indicated as red diamond ($\diamond$). The number of observations within the bins is written on top. Based on 20 high quality fingerprints from FVC2002 DB1.}}
		\label{figure:PoissonRegression}
	\end{figure}
	
	The Poisson regression line in Figure~\ref{figure:PoissonRegression} (red line), which is surprisingly well defined by the means of the massive histograms, confirms that on average the actual minutiae count increases with a slope close to one (95\% confidence interval $[0.903, 1.175]$) with a significant intercept of $0.14$ (95\% confidence interval $[0.106, 0.184]$; $p<10^{-12}$), indicating that the actual number of minutiae is larger than the number of minutiae necessary based on OF and RF divergence (dotted black line). We refer to the additional minutiae as \textit{random minutiae}. 
	
	In the following sections, we investigate the separation of the two minutiae point processes and in Section~\ref{sec:proof} we show that random minutiae can be \textit{characteristic} in the sense that they can provide valuable information for distinguishing fingerprints with similar OFs.
	
	\begin{remark}\label{rem:lambda0}
		Comparing the intercept of $0.14$ random minutiae per patch to the average number of $0.50$ total minutiae per patch, we obtain the rule of thumb that, out of $7$ minutiae, $5$ are necessary and $2$ are random.\\
		The images in FVC2002~DB1 have size $388 \times 374$ pixels (at a resolution of 500 dpi). Using the area of one pixel as our spatial unit (each pixel is a $0.0508 \times 0.0508\,\mathrm{mm}^2$ square) we conclude from the fact that the average patch size is approximately $1451$ pixels that we may use
		\begin{equation}\label{eq:lambda0}
		\lambda_0 \approx 10^{-4}
		\end{equation}
		as an initial estimate for the random minutiae intensity.
	\end{remark}
	
	\section{Modelling Necessary and Random Minutiae}
	\label{sec:model}
	
	We assume that a minutiae pattern $\left\{z_1, z_2,\dots, z_k\right\} \sub \iks $ is a sample of the superposition of two independent point processes $\Xi$ and $\Eta$ modelling the random and the necessary minutiae, respectively; for an introduction to point processes, see e.g.~\cite{moller_statistical_2003}.
	
	It is well known, see e.g.~\cite{stoney_distribution_1988,chen_statistical_2006,gottschlich_separating_2014}, that minutiae cannot be arbitrarily close to one another; they repel each other on a local scale. Indeed, due to the discrete nature of the ridge pattern, we cannot observe minutiae pairs at distance smaller than the inter-ridge distance. Although on good quality fingerprints, upon close inspection, occasionally closer minutiae pairs can be seen, e.g.\ bifurcations with one very short ridge, as these cannot be well discriminated from noise, they are usually removed as \emph{false minutiae}, cf. \cite[p.~157--158]{maltoni_handbook_2009}). This minimal distance effect is well visible in Figure~\ref{figure:PCF} showing the \emph{pair correlation function} (PCF) with approximate pointwise 95\% confidence intervals, estimated from the 20 hand-marked fingerprints considered in Figure~\ref{figure:PoissonRegression} using \cite[Subsections 7.10.2 and 16.8.2]{baddeley_spatial_2015} (adjusting for the inhomogeneous intensities). Intuitively, the PCF shows the ratio between the probability of observing a pair of points at a given distance and the same probability assuming independent occurrence of points. In particular, for a general Poisson process the PCF is constant one (dotted green line in Figure~\ref{figure:PCF}). Inhibition of points leads to values $<1$ and excitation of points to values $>1$. For a precise definition of the PCF, see \cite[Definitions 4.3--4.4]{moller_statistical_2003}.
	
	\begin{figure}[h]
		\centering
		\includegraphics[width = \textwidth]{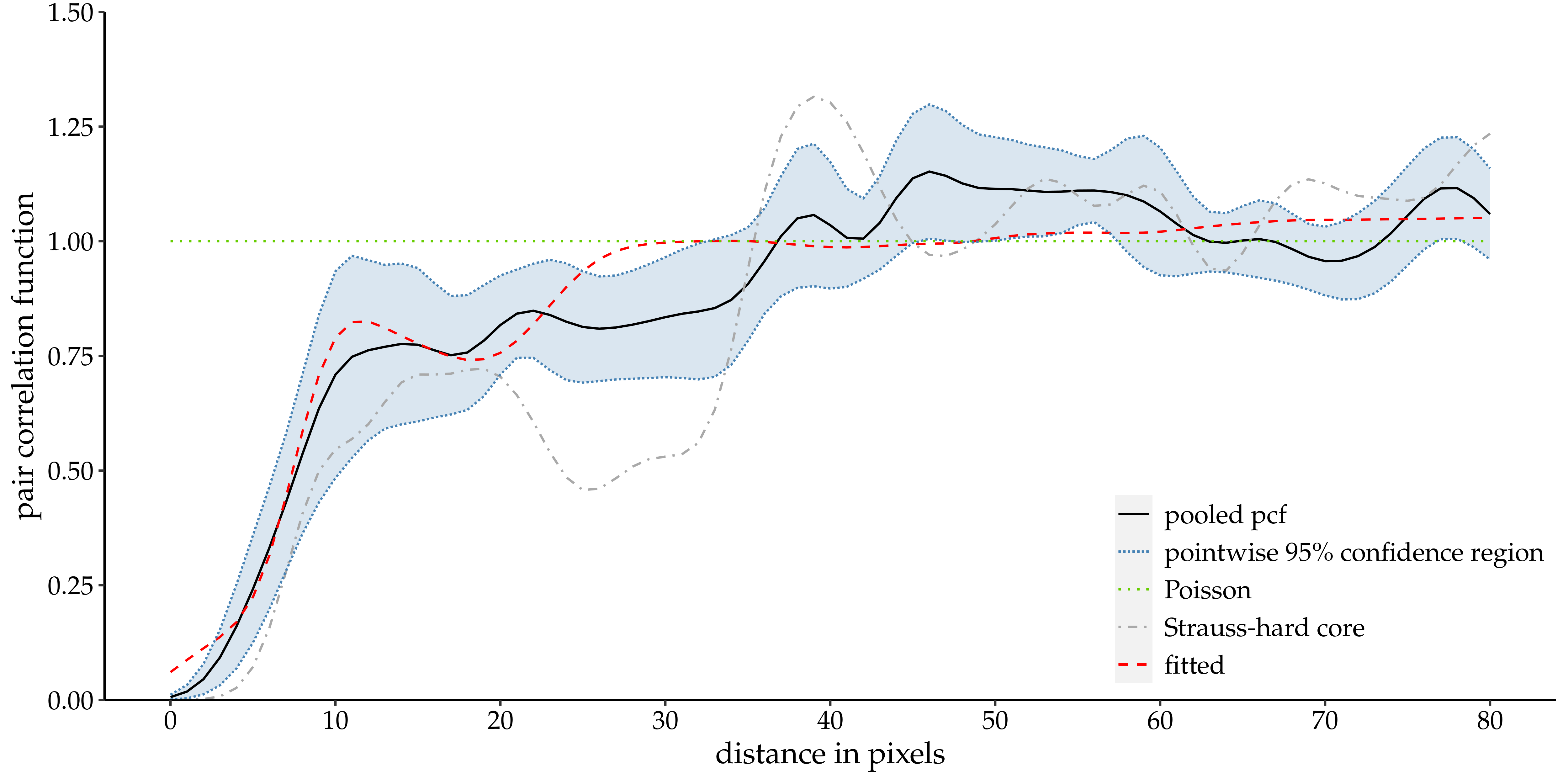}
		\caption{\it Pooled pair correlation function (PCF) based on the 20 
			high quality fingerprint images from FVC2002 DB1 with hand marked 
			minutiae (solid black). Dotted blue shaded: approximate pointwise 95\% 
			confidence intervals based on sample variances. Dotted green: 
			Theoretical PCF under the hypothesis of no interaction. Dash-dotted 
			grey: pooled PCF based on simulated Strauss processes with hard core. 
			Dashed red: pooled PCF for the 20 fitted models using the posterior 
			mean from Section~\ref{sec:test}.}
		\label{figure:PCF}
	\end{figure}
	
	Figure~\ref{figure:PCF} shows roughly two regimes of interaction. A regime of very strong inhibition in the range up to about 5--10 pixels and a regime of moderate inhibition up to 35--40 pixels. This suggests modelling the bulk of the minutiae by a two-scale Strauss process; see \cite[Example~6.2]{moller_statistical_2003}. We choose zero interaction at distances $\leq h$ (hard core, banning points closer than $h$) and interaction $\gamma \in (0,1)$ at distances $\in (h,R]$, and refer to the resulting point process as a \emph{Strauss process with hard core}. For comparison, the pooled PCF estimate is shown in Figure~\ref{figure:PCF} (dash-dotted grey curve) based on 20 simulated Strauss processes with hard core, having activity functions $\beta\,\mu_\ell(z)$, $1 \leq \ell \leq 20$, where $\mu_\ell$ is the necessary minutiae intensity~\eqref{eq:min.estimate} obtained from the $\ell$-th fingerprint image. The hard core distance $h = 8$ was chosen as the average inter-ridge distance (see Section~\ref{sec:mdf}). A pilot study on FVC2002 DB1 suggested that a Strauss \emph{interaction distance} $R$ of approximately three times the average inter-ridge distance seems to be a reasonable choice since only about 6\% of all minutiae pairs of the considered fingerprints have a smaller distance. The parameters $\beta = 1.9$ and $\gamma = 0.37$ are reasonable choices in view of the simulations considered later in Section~\ref{sec:test}. Note that the real minutiae patterns also contain the random minutiae, which essentially explains that the grey curve is visibly too small for values up to approximately $R$. To demonstrate that this deficiency will be suitably corrected once we apply our model in Section~\ref{sec:test} to the same data, we also show the PCF from that model fit (dashed red curve). Clearly our choice of $R$ is somewhat too small, as we will discuss in Section~\ref{sec:test}. 
	
	As a model for the random minutiae process $\Xi$, we choose a homogeneous Poisson process with unknown intensity $\lambda \ge 0$. Such a process has density
	\begin{align*}
	f_\lambda: \mathfrak{N} \to [0,\infty), \qquad f_\lambda(\xi) = \e^{(1-\lambda)|\iks|}\, \lambda^{n(\xi)}
	\end{align*}
	w.r.t.\ the standard Poisson process (homogeneous Poisson process with intensity $1$), see \cite[Proposition 3.8]{moller_statistical_2003}. Here $\mathfrak{N}$ denotes the set of all finite point configurations of $\iks$ and $n(\xi)$ is the number of points in~$\xi$.
	Ignoring that, in practice, random minutiae cannot be closer than the inter-ridge distance is harmless as their intensity is rather low, both in absolute terms and compared to the intensity of necessary minutiae (see Remark \ref{rem:lambda0}).
	
	If we also assumed that the necessary minutiae process $\Eta$ was a Poisson process, but inhomogeneous with intensity proportional to (\ref{eq:min.estimate}), then $\Xi \dcup \Eta$ would also be Poisson distributed, cf.~\cite[Proposition 3.6]{moller_statistical_2003}. Under this assumption the theoretical pair correlation function would be one which, with regard to Figure~\ref{figure:PCF}, contrasts reality.
	
	Last but not least, modelling the necessary minutiae as a point process with substantial inhibition of points is also advantageous from a conceptual point of view: up to certain errors arising from the discretization of the OF and RF into minutiae information as well as from data acquisition and processing, the necessary minutiae counts should be \emph{determined} by the underlying necessary minutiae intensity. Some inhibition between points is required to keep the variances of minutiae counts in regions with high necessary minutiae intensity small enough to be compatible with the data. Simulations we performed (not shown here) indicate that inhomogeneous Poisson processes based on the same intensity have much too high variances.
	
	The Strauss process with hard core has density
	\begin{align*}
	g_{\beta,\gamma}: \mathfrak{N}\to [0,\infty), \quad g_{\beta,\gamma}(\eta) = \alpha\, \left(\prod_{z \in \eta }  \beta(z)\right)\, \gamma^{s_R(\eta)}\, \ind\left( \dmin(\eta) > h \right)
	\end{align*}
	w.r.t.\ to the standard Poisson process. Here $h > 0$ is the hard core distance,
	\begin{align*}
	\dmin(\eta) = \min \{\|z-w\|: \{z,w\} \sub \eta\}
	\end{align*}
	is the minimum inter-point distance in $\eta$ (by convention $\{z,w\} \sub \eta$ shall always exclude the case $z=w$) and	
	\begin{align*}
	s_R(\eta) = \sum_{\{z,w\}\sub \eta} \ind\left(\|z-w\|\le R \right)
	\end{align*}
	is the number of pairs of points that lie within the interaction distance $R > h$ of each other. The \emph{activity} or \emph{trend} function $\beta: \iks \to [0,\infty)$ governs the intensity and takes up OF divergence and RF divergence and $\gamma \in (0,1)$ is the repulsion strength at distances $\in (h,R]$, meaning each point pair at distance in $(h, R]$ is penalized by a factor $\gamma<1$. Based on formula~(\ref{eq:min.estimate}), we assume 	
	\begin{equation}
	\label{eq:trendmod}
	\beta(z) = \beta\cdot\mu(z) \quad \mbox{ with } \quad \mu(z) =\left|\Phi(z)\div \vec F(z)+ \left\langle \nabla\Phi(z), \vec F(z) \right\rangle \right|
	\end{equation}
	for some factor $\beta \approx 1$. The factor $\alpha = \alpha(\beta,\gamma)$ denotes the normalising constant of the probability density $g_{\beta,\gamma}$ and is intractable, cf.~\cite[Section 6.2]{moller_statistical_2003}. We expect that $\beta$ is in fact larger than one because the presence of repulsion requires an activity larger than a Poisson intensity yielding comparable number of observed points, e.g.
	\cite{jensen2001review,eckel2009modelling}. Even in the homogeneous case it is not clear, how $\beta$ and $\gamma$ interact, \cite{baddeley2012fast,coeurjolly2018intensity} give some approximations.

	The interaction distances $h$ and $R$ are assumed to be known in advance, since joint estimation of $\gamma, h, R$ simultaneously is notoriously difficult due to strong negative correlation, cf.~\cite{redenbach_classification_2015}. We choose the hard core distance $h$ as the average inter-ridge distance of the finger -- which seems to be be fairly realistic -- and the interaction distance $R$ to be three times as large, see above.
	
	The observed minutiae pattern $\zeta=\{z_1,z_2,\dots,z_k\}$ can then be written as $\zeta=\xi \dcup \eta$ where $\xi$ and $\eta$ are realizations of $\Xi$ and $\Eta$, respectively. We introduce a latent variable $\mathbf{W}\in \{0,1\}^k$ where $\mathbf{W}_i = \ind\left( z_i \in \eta \right)$, so that $\mathbf{W}_i=1$ means that minutia $z_i$ is necessary. 
	We combine the parameters into a vector $\boldsymbol{\theta} = (\lambda, \beta, \gamma) \in \Theta := [0, \infty) \times [0, \infty)\times[0,1]$. Then, by independence of $\Xi$ and $\Eta$, the density of $\Xi \dcup \Eta$ given $\mathbf{W}$ is given by
	\begin{align}
	\label{eq:likelihood}
	f_\lambda(\xi)\, g_{\beta,\gamma}(\eta).
	\end{align}
	The question arises how to find suitable values of $\boldsymbol{\theta}$ and an assignment $\mathbf{W}$ of the minutiae to $\xi$ and $\eta$. The computation of a maximum likelihood estimator for the parameters is notoriously difficult due to the intractable normalising constant. A maximum pseudo-likelihood approach for the superimposed processes is studied in \cite{wieditz_separating_2021}. While it is computationally expensive, it does not provide information about $\mathbf{W}$. Even more importantly, it is questionable whether there is only one single choice of $(\boldsymbol{\theta}, \mathbf{W})$ that fits best. Our view is that there likely are several choices which all fit reasonably well, particularly for $\mathbf{W}$. Bearing this in mind, we adopt a Bayesian approach, exploring the posterior distribution of $(\boldsymbol{\theta}, \mathbf{W})$ given the minutiae point pattern~$\xi \dcup \eta$. Not only does this yield information about the parameter values, it also provides a quantification of the uncertainty of the assignment to the classes of necessary and random minutiae. 	
	
	\section{Bayesian Inference using MCMC}	\label{sec:mcmc}
	For better readability we do not distinguish between random variables and their realizations when writing ``$\boldsymbol{\theta}$'' or ``$\mathbf{W}$''. Given a minutiae pattern $\zeta = \left\{ z_1,z_2,\dots, z_k \right\} \sub \iks$, Bayes' theorem yields for the posterior distribution 
	\begin{align}
	\label{eq:Bayes}
	\pi(\boldsymbol{\theta}, \mathbf{W} \mid \zeta) = \frac{\pi(\zeta \mid \boldsymbol{\theta}, \mathbf{W}) \, \pi(\boldsymbol{\theta}, \mathbf{W})}{\int \pi(\zeta \mid \tilde{\boldsymbol{\theta}}, \tilde{\mathbf{W}}) \, \pi(\tilde{\boldsymbol{\theta}}, \tilde{\mathbf{W}}) \d (\tilde{\boldsymbol{\theta}}, \tilde{\mathbf{W}}) },
	\end{align}
	where $\pi(\zeta \mid \boldsymbol{\theta}, \mathbf{W})$ is the \textit{likelihood} of our data given parameters and labels as described in (\ref{eq:likelihood}) and $\pi(\boldsymbol{\theta}, \mathbf{W})$ models our \textit{prior} belief about the parameters and labels. We face a doubly intractable problem because the denominator in (\ref{eq:Bayes}) is intractable, and $\pi(\zeta\mid \boldsymbol{\theta}, \mathbf{W})$ in the numerator contains another intractable normalising constant. Ignoring the second factor for the moment, we can eliminate the first intractability by applying a Metropolis--Hastings algorithm,	see Algorithm~\ref{algo:MCMC}, which produces samples from the posterior distribution~\eqref{eq:Bayes}. For an overview about MCMC methods see e.g.\ \cite{brooks_handbook_2011} or \cite{liu_monte_2004}.
	
	
	\begin{algo}[Framework for the Minutiae Separation Algorithm (MiSeal)] 
		$ $\\[.5mm]
		\label{algo:MCMC}
		\begin{algorithm*}[H]
			\DontPrintSemicolon
			\SetKwInOut{Input}{Input}\SetKwInOut{Output}{Output}
			\SetNlSkip{0em}
			\SetInd{.5em}{.25em}
			\BlankLine
			\Indentp{-.5em}
			\Input{\hspace{1ex}Minutiae pattern $\zeta = \left\{ z_1,z_2,\dots, z_k \right\} \sub \iks$. }
			\Indentp{1em}
			\BlankLine
			Choose some initial $(\boldsymbol{\theta}^{(0)}, \mathbf{W}^{(0)}) \in \Theta \times \{0,1\}^k$. \;
			\For{ $t=0,1,2,\dots$ }{
				
				Given $(\boldsymbol{\theta}^{(t)}, \mathbf{W}^{(t)}) = (\boldsymbol{\theta}, \mathbf{W})$, generate a candidate $(\boldsymbol{\theta}', \mathbf{W}')$ for the next sample from the probability density $q(\boldsymbol{\theta}', \mathbf{W}' \mid \boldsymbol{\theta}, \mathbf{W})$.\;
				Calculate the \textit{Hastings ratio} 
				\begin{align}
				\label{eq:Hratio}
				H(\boldsymbol{\theta}', \mathbf{W}' \mid \boldsymbol{\theta}, \mathbf{W}) &= \frac{\pi(\boldsymbol{\theta}', \mathbf{W}' \mid \zeta)}{\pi(\boldsymbol{\theta}, \mathbf{W} \mid \zeta)} \, \frac{q(\boldsymbol{\theta}, \mathbf{W} \mid \boldsymbol{\theta}', \mathbf{W}')}{q(\boldsymbol{\theta}', \mathbf{W}' \mid \boldsymbol{\theta}, \mathbf{W})}. 
				\end{align}\;
				\vspace{-4ex}
				Accept the candidate with probability $\min\{H(\boldsymbol{\theta}', \mathbf{W}'\mid \boldsymbol{\theta}, \mathbf{W}), 1\}$. \;
				In case of acceptance set $(\boldsymbol{\theta}^{(t+1)}, \mathbf{W}^{(t+1)}) = (\boldsymbol{\theta}', \mathbf{W}')$, otherwise $(\boldsymbol{\theta}^{(t+1)}, \mathbf{W}^{(t+1)}) = (\boldsymbol{\theta}, \mathbf{W})$.
			}
			If deemed necessary, discard the first $t_0$ samples (burn-in).
			\BlankLine
			\Indentp{-1em}
			\Output{\hspace{1ex}A sample $(\boldsymbol{\theta}^{(t)}, \mathbf{W}^{(t)})_{t=t_0+1,t_0+2,\dots}$ from the distribution induced by $\pi(\boldsymbol{\theta}, \mathbf{W}\mid \zeta)$.}
			\BlankLine
			\BlankLine
		\end{algorithm*}
	\end{algo}

	\begin{remark}
		The performance of Algorith~\ref{algo:MCMC} crucially depends on the computation of the Hastings ratio (\ref{eq:Hratio}). Applying (\ref{eq:Bayes}), this can be written as	
		\begin{align}
		H(\boldsymbol{\theta}', \mathbf{W}' \mid \boldsymbol{\theta}, \mathbf{W}) &= \frac{\pi(\zeta\mid \boldsymbol{\theta}', \mathbf{W}')}{\pi(\zeta \mid \boldsymbol{\theta}, \mathbf{W})} \frac{\pi(\boldsymbol{\theta}', \mathbf{W}')}{\pi(\boldsymbol{\theta}, \mathbf{W})} \frac{q(\boldsymbol{\theta}, \mathbf{W}\mid \boldsymbol{\theta}',\mathbf{W}')}{q(\boldsymbol{\theta}',\mathbf{W}' \mid \boldsymbol{\theta}, \mathbf{W})}.
		\end{align}
		In the following sections we elaborate further details on how to compute the individual quotients, the choice of our priors as well as the update procedures for $\boldsymbol{\theta}$ and $\mathbf{W}$. 
		
		For the updates, we employ a random scan Gibbs sampler \cite{liu_monte_2004} with update probabilities $p_{\boldsymbol{\theta}}$ for $\boldsymbol{\theta}$ and $1-p_{\boldsymbol{\theta}}$ for $\mathbf{W}$ because the computation of the Hastings ratio for a joint update turned out to be difficult. A value of $p_{\boldsymbol{\theta}} = 0.05$ (reflecting on average 19 proposed flips out of the approx.\ 30--60 minutiae per finger for each $\boldsymbol{\theta}$ update) approximately yields the fastest mixing. The update for each component employs a Markov chain yielding a variant from the \textit{Metropolis-within-Gibbs} class of algorithms, see \cite{roberts2006harris}. 
	\end{remark}
	
	\subsection{Choice of priors}
	\label{sec:priors}
	
	We write $\pi(\boldsymbol{\theta}, \mathbf{W}) = \pi(\boldsymbol{\theta})\pi(\mathbf{W})$ and assume the prior of the parameters $\pi(\boldsymbol{\theta}) = \pi(\lambda)\pi(\beta)\pi(\gamma)$ is a product of the priors of the single parameters chosen as follows.
	
	We choose $\beta \sim \Gamma(a_1, b_1)$ where the parameters $a_1=b_1=5$ are chosen such that the expected value of $\beta$ equals one (see around \eqref{eq:trendmod}) and the variance is reasonably large. For flexibility we choose $\gamma \sim \Beta(p_1,q_1)$ to be beta distributed with $p_1=2$, $q_1=5$ (hence, $\E \gamma = \frac27$), since we expect the process to be rather inhibitive also outside the hard core distance. 
	
	For the intensity $\lambda$ of the random minutiae we choose the prior to be a Gamma distribution $\Gamma(a_0,b_0)$. Being conjugate to the Poisson likelihood, this prior has the advantage that we do not have to perform Hastings steps when updating $\lambda$, but can draw directly from the posterior distribution conditional on $\mathbf{W},\beta$ and $\gamma$, cf.~(\ref{eq:18}) below. We choose the parameters $a_0=5$, $b_0 = \frac{5}{\lambda_0}$ such that the  expected value of the prior is $\lambda_0=10^{-4}$ from (\ref{eq:lambda0}) and its variance is $\lambda^2_0/5$ such that the ratio $1/\sqrt{5}$ of standard deviation over mean is reasonably sized.
	
	Furthermore, for every imprint featuring $k$ minutiae in its region of interest $\iks$, Remark \ref{rem:lambda0} yields an expected number $\lambda_0|\iks|$ of random minutiae in $\iks$. Hence, for each finger individually, we choose the prior for the label vector as $\pi(\mathbf{W}) = \bigotimes_{i=1}^k \pi(\mathbf{W}_i)$ with $\mathbf{W}_i \iid \mathrm{Ber}\left( p_\mathbf{W} \right)$ and $p_\mathbf{W} = \max\left\{ 1-\frac{\lambda_0 |\iks|}{k}, 0\right\}$. We discard infeasible label vectors, i.e.\ label vectors for which some pairs of $\{z_i \mid \mathbf{W}_i = 1\}$  have distance smaller than $h$ and thereby violate the hard core condition. The ratio of priors for different parameters $(\boldsymbol{\theta}, \mathbf{W})$, $(\boldsymbol{\theta}', \mathbf{W}')$ thus computes as	
	\begin{align}
		\notag &\frac{\Pi(\theta', W')}{\Pi(\theta, W)} = \frac{\Pi(\theta')}{\Pi(\theta)} \frac{\Pi(W')}{\Pi(W)}\\
		\label{eq:prior}&= \left( \frac{\lambda'}{\lambda} \right)^{a_0-1} \e^{-b_0(\lambda'-\lambda)} \, \left( \frac{\beta'}{\beta} \right)^{a_1 - 1} \e^{-b_1(\beta'-\beta)}\, \left(\frac{\gamma'}{\gamma}\right)^{p_1-1} \left(\frac{1-\gamma'}{1-\gamma}\right)^{q_1 - 1}\, \left( \frac{p_W}{1-p_W} \right)^{\ell'-\ell},
		\end{align}
	where $1\le \ell, \ell'\le k$ denote the number of ones in $\mathbf{W}, \mathbf{W}'$, respectively.
	
	We have adjusted the variances of the priors such that they concentrate on a domain we deem reasonable according to a pilot study (not shown here). We keep them rather uninformative, however, to avoid undesirable dependence of the posterior on our particular prior choices.
	
	\subsection{Update of $\boldsymbol{\theta}$}
	\label{sec:update.theta}
	
	When updating $\boldsymbol{\theta}$, we randomly choose to update either $\lambda$ or $(\beta, \gamma)$. To this end, toss a coin with success probability $p_\lambda = 0.2$  (this is slightly smaller than $1/3$, taking into account that, due to the explicitly available posterior, see below, there are no rejects for the $\lambda$-updates). In case of success, update $\lambda$, which, since the Gamma prior is conjugate for the Poisson likelihood, we can draw directly from the posterior distribution, namely from	
	\begin{align}
	\label{eq:18}
	\lambda \mid (\zeta, \mathbf{W}) \sim \Gamma( a_0 + n_0, b_0 + |\iks| ),
	\end{align}
	where $n_0$ is the number of minutiae currently labelled as random in $\mathbf{W}$. Note that this does not depend on any of the other parameters of $\boldsymbol{\theta}$ nor on the current state of~$\lambda$.
	
	In case of failure, we update the parameters of the Strauss process by proposing a normally distributed step in the natural parameter space ($\log$-space), i.e.\ our proposal $(\beta',\gamma')$ is log-normally distributed,	
	\begin{align}
	\label{eq:19}
	\begin{pmatrix}
	\beta'\\ \gamma'
	\end{pmatrix} \sim \mathcal{LN} \left( \begin{pmatrix}
	\log \beta\\ \log \gamma 
	\end{pmatrix}, \begin{pmatrix}
	\sigma_1^2 & \rho_{12} \sigma_1\sigma_2 \\
	\rho_{12} \sigma_1\sigma_2 & \sigma_2^2
	\end{pmatrix} \right), 
	\end{align}
	where the parameters $\sigma_1$, $\sigma_2$, $\rho_{12}$ are fixed (see end of this subsection). Denote the covariance matrix in (\ref{eq:19}) by ${\Sigma}$. Then, the proposal density for a $(\beta,\gamma)$-update from $\boldsymbol{\theta} = (\lambda, \beta, \gamma)$ to $\boldsymbol{\theta}' = (\lambda, \beta', \gamma')$ is given as	
	\begin{align*}
	q\left(\boldsymbol{\theta}',\mathbf{W} \mid \boldsymbol{\theta},\mathbf{W} \right)
	&= \frac{1}{2\pi \sqrt{\det {\Sigma}}}\, \frac{1}{\beta'\gamma'} \exp \left( -\frac12  \begin{pmatrix}
	\log \beta'/\beta\\ \log \gamma'/\gamma
	\end{pmatrix} 
	^\top {\Sigma}^{-1} 
	\begin{pmatrix}
	\log \beta'/\beta\\ \log \gamma'/\gamma
	\end{pmatrix} 
	\right)
	\end{align*}
	and hence	
	\begin{align}
	\label{eq:proposal_theta}
	\frac{q(\boldsymbol{\theta},\mathbf{W} \mid \boldsymbol{\theta}',\mathbf{W})}{q(\boldsymbol{\theta}',\mathbf{W}\mid \boldsymbol{\theta},\mathbf{W})} &= \frac{\beta' \gamma'}{\beta \gamma}.
	\end{align}	
	To compute the Hastings ratio for the $\boldsymbol{\theta}$-update, we consider the likelihood ratio	
	\begin{align} 
	\frac{\pi(\zeta\mid \boldsymbol{\theta}', \mathbf{W})}{\pi(\zeta\mid \boldsymbol{\theta}, \mathbf{W})} &= \frac{\alpha(\beta',\gamma')}{\alpha(\beta,\gamma)} \, \left( \frac{\beta'}{\beta} \right)^{n(\eta)} \, \left( \frac{\gamma'}{\gamma} \right)^{s_R(\eta)}.
	\end{align}
	Note that this ratio still contains a ratio of normalising constants which cannot be computed explicitly. To overcome this problem, we apply the auxiliary variable method which goes back to \cite{besag_spatial_1993}. For application in point processes we refer to~\cite{berthelsen_bayesian_2006,murray_mcmc_2012,redenbach_classification_2015,rajala_variational_2016} and the references therein. 
	
	To this end, we extend the state space and introduce an auxiliary point pattern $\tilde \chi$ with density $\phi(\tilde \chi \mid \mathbf{W}, \zeta)$ w.r.t.\ the standard Poisson process which does not depend on the current $\boldsymbol{\theta}$. The point pattern $\tilde \chi$ is then included in the model as an additional variable. For this, we have to define a new proposal distribution on the extended state space, which we choose as	
	\begin{align*}
	\tilde q(\boldsymbol{\theta}', \mathbf{W}, \tilde \chi' \mid \boldsymbol{\theta}, \mathbf{W}, \tilde \chi) = \tilde q(\boldsymbol{\theta}', \mathbf{W}, \tilde \chi' \mid \boldsymbol{\theta}, \mathbf{W}) = g_{\beta',\gamma'}(\tilde \chi') \, q(\boldsymbol{\theta}', \mathbf{W} \mid \boldsymbol{\theta}, \mathbf{W}),
	\end{align*}
	i.e.\ we draw the new auxiliary point pattern $\tilde \chi'$ independently of the current auxiliary point pattern $\tilde \chi$ as a realization of a Strauss process with hard core having parameter $\boldsymbol{\theta}'$, whereas the proposal for the parameter $\boldsymbol{\theta}'$ remains as before. Then, the Hastings ratio for a parameter update from $\boldsymbol{\theta} = (\lambda, \beta, \gamma)$ to $\boldsymbol{\theta}' = (\lambda, \beta', \gamma')$ is 	
	\begin{align}
	\notag H(\boldsymbol{\theta}', \mathbf{W}, \tilde \chi' \mid \boldsymbol{\theta}, \mathbf{W}, \tilde \chi) &= \frac{\pi(\boldsymbol{\theta}', \mathbf{W}, \tilde \chi' \mid \zeta)}{\pi(\boldsymbol{\theta}, \mathbf{W}, \tilde \chi \mid \zeta)}\, \frac{\tilde q(\boldsymbol{\theta}, \mathbf{W}, \tilde \chi \mid \boldsymbol{\theta}', \mathbf{W}, \tilde \chi')}{\tilde q(\boldsymbol{\theta}', \mathbf{W}, \tilde \chi' \mid \boldsymbol{\theta}, \mathbf{W}, \tilde \chi)}\\[1mm]
	\label{eq:20}&= \frac{\phi(\tilde \chi' \mid \mathbf{W}, \zeta)}{\phi(\tilde \chi \mid \mathbf{W}, \zeta)} \, \frac{g_{\beta',\gamma'}(\eta)}{g_{\beta,\gamma}(\eta)} \, \frac{\pi(\boldsymbol{\theta}', \mathbf{W})}{\pi(\boldsymbol{\theta}, \mathbf{W})}\, \frac{g_{\beta,\gamma}(\tilde \chi)}{g_{\beta',\gamma'}(\tilde \chi')} \frac{q(\boldsymbol{\theta}, \mathbf{W} \mid \boldsymbol{\theta}', \mathbf{W})}{q(\boldsymbol{\theta}', \mathbf{W} \mid \boldsymbol{\theta}, \mathbf{W})}
	\end{align}
	where the ratio of priors and proposals can be obtained from (\ref{eq:prior}) and (\ref{eq:proposal_theta}), respectively. Since the normalising constants of the two $g_{\beta,\gamma}$-terms and the two $g_{\beta', \gamma'}$-terms cancel, this Hastings ratio can be computed explicitly. However, in every update we have to draw a new point pattern from a Strauss process with hard core, for which we again have to run a Markov chain. This requires a considerable additional effort in each update step for $(\beta, \gamma)$. 
	
	A crucial influence on the algorithm's performance is the choice of $\phi$ which ideally should fit well to the proposal density $g_{\beta,\gamma}$. The best choice would of course be to choose $\phi(\tilde \chi \mid \mathbf{W}, \zeta) = g_{\beta,\gamma}(\tilde \chi)$, which is not feasible since then~(\ref{eq:20}) would contain the normalising constants again. In 	\cite{redenbach_classification_2015}, the density of a Poisson process was used, which results in rather poor mixing behaviour, cf.~\cite[Figure 6]{redenbach_classification_2015}. We therefore choose for $\phi$ the density of another Strauss process with hard core, fixing its parameter at 	$\hat{\boldsymbol{\theta}} = \hat{\boldsymbol{\theta}}_{\text{MPLE}}(\zeta,\mathbf{W})$, where $\mathbf{W}$ is the 
	current label vector and $\hat{\boldsymbol{\theta}}_{\text{MPLE}}(\zeta,\mathbf{W})$ is the maximum pseudo-likelihood estimate (MPLE) based on the minutiae currently labelled as necessary; see e.g.~\cite[Section~13.13]{baddeley_spatial_2015} and the references given there. If we knew the true $\mathbf{W}$ in advance, this would be a good initial guess for the parameters. However, in practice the true $\mathbf{W}$ is not known. We therefore adapt $\hat{\boldsymbol{\theta}}$ iteratively during burn-in and keep it fixed for the rest of the run, so that we still obtain convergence to the desired posterior distribution.
	
	For the proposed MCMC algorithm the proposal variances $\sigma_1$, $\sigma_2$ and correlation $\rho_{12}$ have to be determined. To this end, we considered for reasonable values of $\sigma_1, \sigma_2$ and $\rho_{12}=0$ a pilot sample and estimate the correlation coefficient $\rho_{12}$ as its sample correlation. The corresponding proposal variances are then adjusted such that the acceptance rate for a parameter proposal is about $23.4\%$ which is optimal according to~\cite{roberts_weak_1997}. For our computations we set $\sigma_1 = 0.07$, $\sigma_2 = 0.05$ and $\rho_{12} = -0.7$.
	
	\subsection{Update of $\mathbf{W}$}
	\label{sec:update.W}
	
	\noindent
	When updating $\mathbf{W}$, we pick one component of $\mathbf{W}$ uniformly at random, $\mathbf{W}_i$, say, and propose to flip it to $\mathbf{W}_i' = 1-\mathbf{W}_i$, while keeping the other components unchanged, $\mathbf{W}_j' = \mathbf{W}_j$ for $j \neq i$. Writing $\zeta = \xi \dcup \eta = \xi' \dcup \eta'$ for the partition in random and necessary minutiae before and after the proposed flip, respectively, we either have $\xi' = \xi \cup \{z_i\}$, $\eta' = \eta \setminus \{z_i\}$ if the flip of $\mathbf{W}_i$ is from $1$ to $0$ or $\xi' = \xi \setminus \{z_i\}$, $\eta' = \eta \cup \{z_i\}$ if the flip is from $0$ to $1$. Thus, the Hastings ratio is given as	
	\begin{align}
	\notag \hspace*{-2mm} &H(\boldsymbol{\theta}, \mathbf{W}', \tilde \chi \mid \boldsymbol{\theta}, \mathbf{W}, \tilde \chi)\\ 
	\notag &= \frac{\pi(\boldsymbol{\theta}, \mathbf{W}', \tilde \chi \mid \zeta)}{\pi(\boldsymbol{\theta}, \mathbf{W}, \tilde \chi \mid \zeta)}\, \frac{\tilde q(\boldsymbol{\theta}, \mathbf{W}, \tilde \chi \mid \boldsymbol{\theta}, \mathbf{W}', \tilde \chi)}{\tilde q(\boldsymbol{\theta}, \mathbf{W}', \tilde \chi \mid \boldsymbol{\theta}, \mathbf{W}, \tilde \chi)} \\[1mm]
	&= \label{eq:HratioW} \frac{\phi(\tilde \chi \mid \mathbf{W}', \zeta)}{\phi(\tilde \chi \mid \mathbf{W}, \zeta)} \, \frac{f_\lambda(\xi')g_{\beta,\gamma}(\eta')}{f_\lambda(\xi) g_{\beta,\gamma}(\eta)} \, \frac{\pi(\boldsymbol{\theta}, \mathbf{W}')}{\pi(\boldsymbol{\theta}, \mathbf{W})}\, \frac{q(\boldsymbol{\theta}, \mathbf{W} \mid \boldsymbol{\theta}, \mathbf{W}')}{q(\boldsymbol{\theta}, \mathbf{W}' \mid \boldsymbol{\theta}, \mathbf{W})} \\[1mm]
	\notag &= \begin{cases}
	\displaystyle \frac{\lambda}{\beta(z_i) \gamma^{t_R(z_i, \eta\setminus\{z_i\})}}\, \frac{1-p_\mathbf{W}}{p_\mathbf{W}} &\text{ if $\mathbf{W}_i=1$ and $\mathbf{W}'_i=0$},\\[2ex]
	\displaystyle \frac{\beta (z_i) \gamma^{t_R(z_i, \eta)}}{\lambda} \, \ind\left( \dmin(\eta \cup \{z_i\}) > h \right) \, \frac{p_\mathbf{W}}{1-p_\mathbf{W}} &\text{ if $\mathbf{W}_i=0$ and $\mathbf{W}'_i=1$}.
	\end{cases}
	\end{align}
	Here, $t_R(z_i, \eta)$ denotes the number of $R$-close neighbours of $z_i$ in $\eta$. Note that in the last equality we used that the first factor of \eqref{eq:HratioW} is equal to~1, which is valid \emph{after} the burn-in phase, when we do not update the parameters of the auxiliary target density $\varphi$ anymore. By contrast, it seems that the same factor was erroneously omitted in \cite[Section~3.2.2]{redenbach_classification_2015}. It is not equal to $1$ there, because the auxiliary target density is homogeneous Poisson with intensity depending on $\eta$ (in our notation). This may be another part of the reason for the unfavourable mixing behaviour in \cite{redenbach_classification_2015}. 
	
	\section{Performance of the Minutiae Separating Algorithm}
	\label{sec:test}
	
	In order to benchmark MiSeal for separating random from necessary minutiae, we first simulate a test scenario close to real fingerprints with true parameters known. To this end, for each of the manually marked 20 fingerprints, we compute a smoothed necessary minutiae intensity function building on $\mu$ from (\ref{eq:min.estimate}), draw  ``true'' parameters $(\lambda, \beta, \gamma)$ from the priors specified in Section~\ref{sec:priors} and simulate a sample of superimposed random and necessary minutiae following the model (\ref{eq:likelihood}). Such a simulated minutiae pattern is depicted in Figure~\ref{figure:04b}; the original minutiae pattern of the same print is seen in Figure~\ref{figure:04a} and the heat map of the necessary minutiae intensity on which the simulation is based in Figure~\ref{figure:04c}.
	
	\begin{figure}[p]
		\centering
		\hspace{-13mm}
		\begin{subfigure}[t]{.30\textwidth}
			\centering
			\includegraphics[height=.4\textheight]{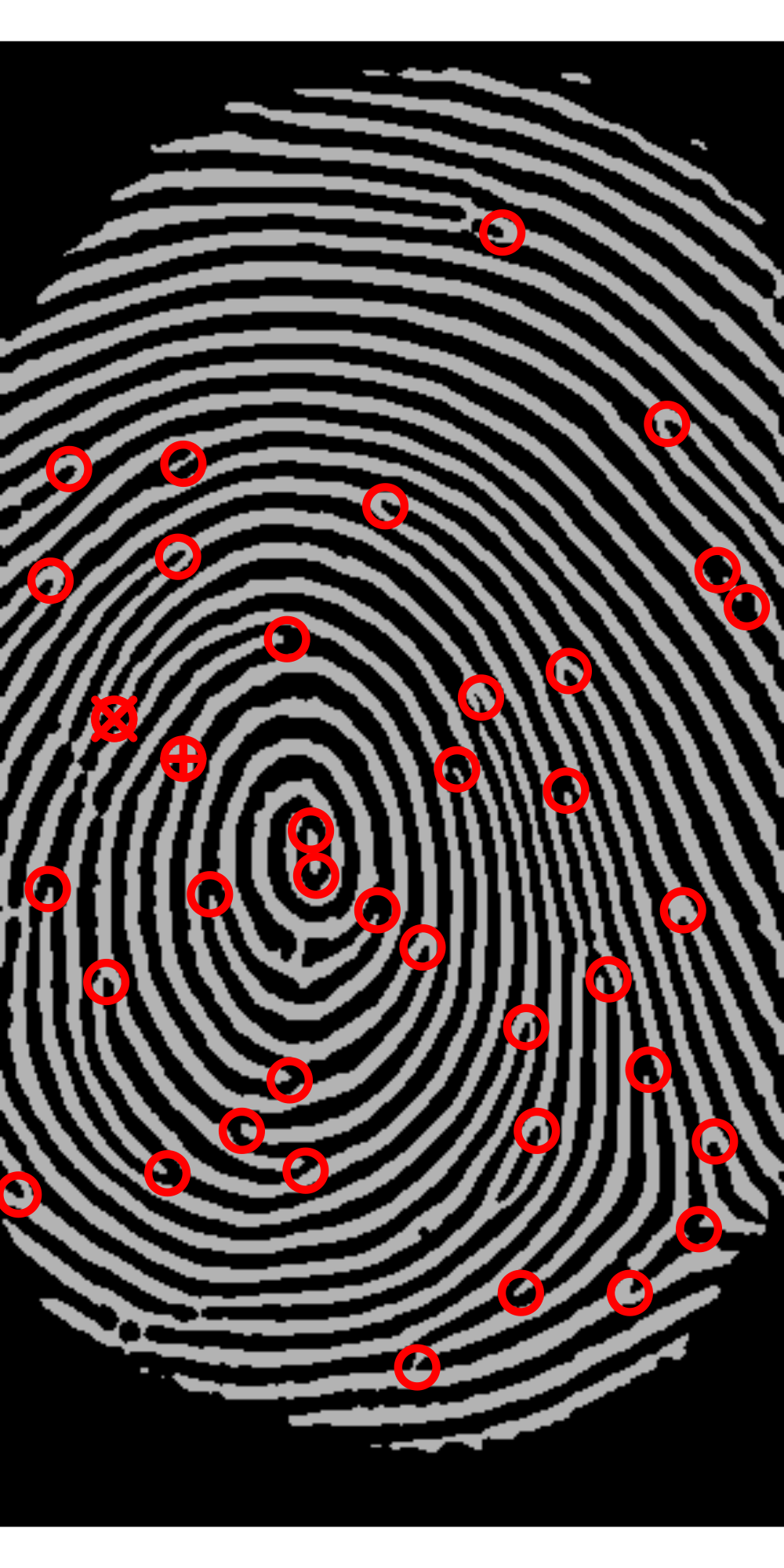}
			\caption{}
			\label{figure:04a}
		\end{subfigure}%
		%
		\begin{subfigure}[t]{.30\textwidth}
			\centering
			\includegraphics[height=.4\textheight]{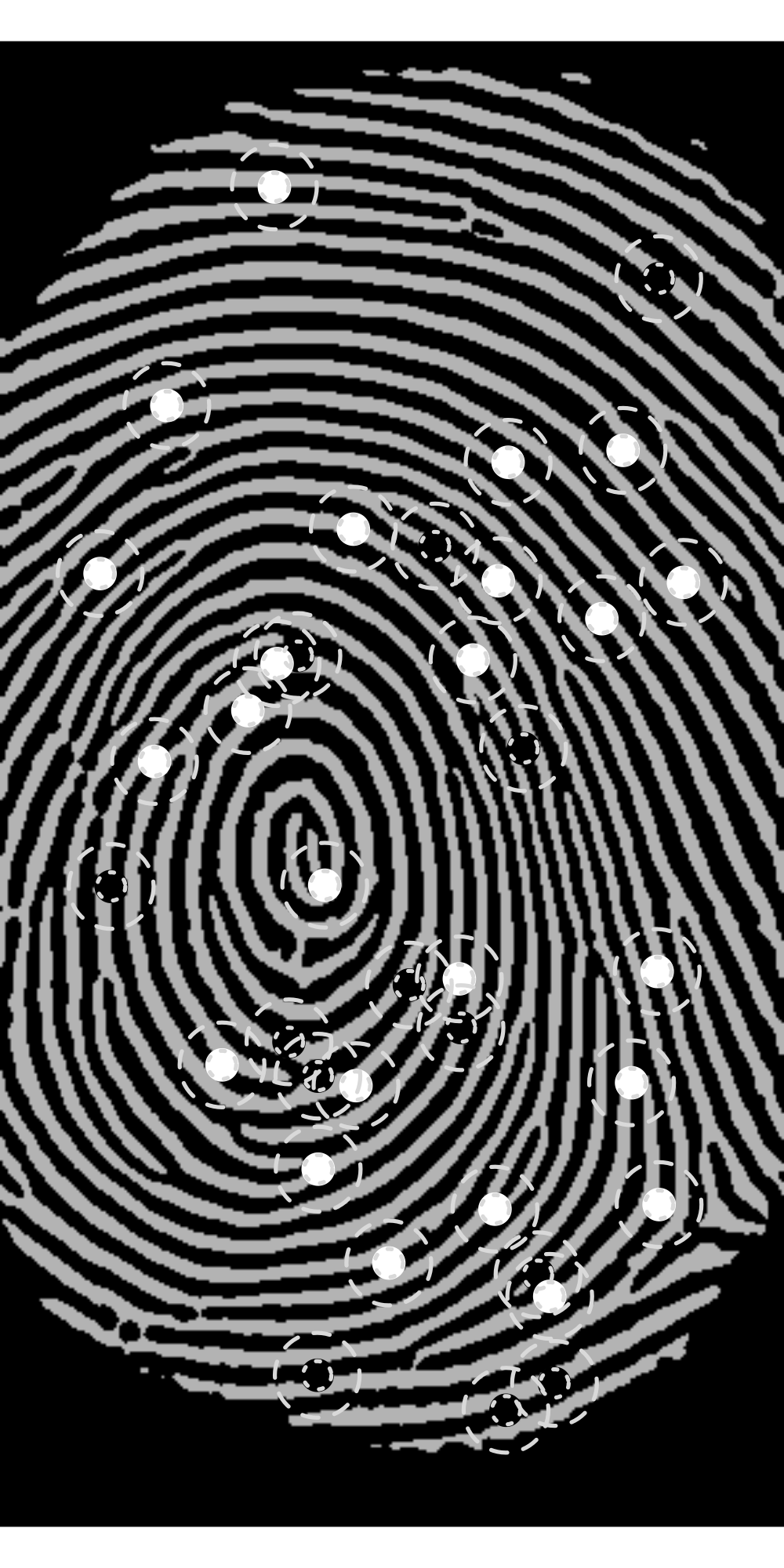}
			\caption{}
			\label{figure:04b}
		\end{subfigure}%
		%
		\hspace{-3.5mm}
		\begin{subfigure}[t]{.35\textwidth}
			\centering	
			\includegraphics[height=.4\textheight]{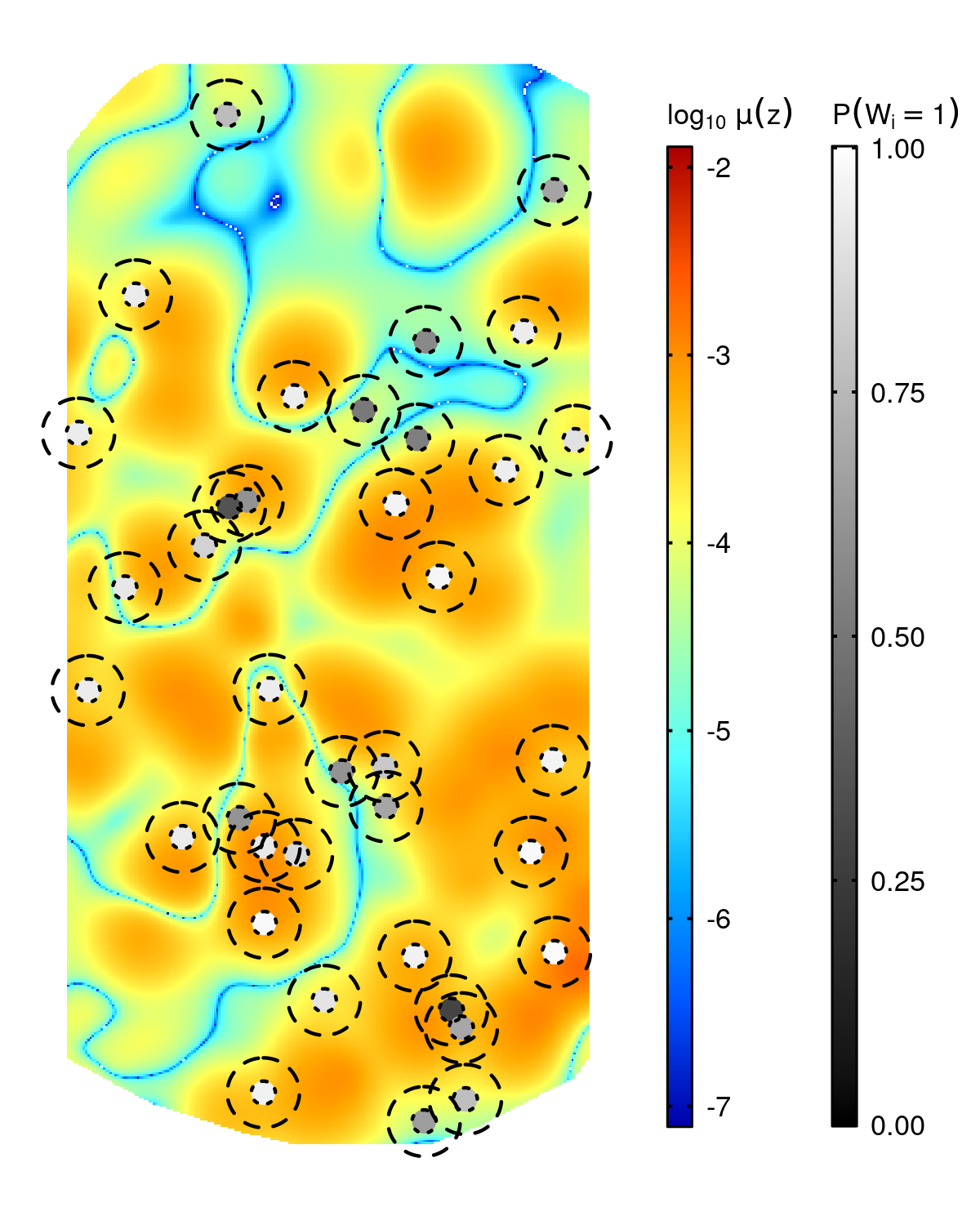}
			\caption{}
			\label{figure:04c}
		\end{subfigure}
		~\\[1ex]
		\begin{subfigure}[t]{.49\textwidth}
			\centering
			\includegraphics[width=\textwidth]{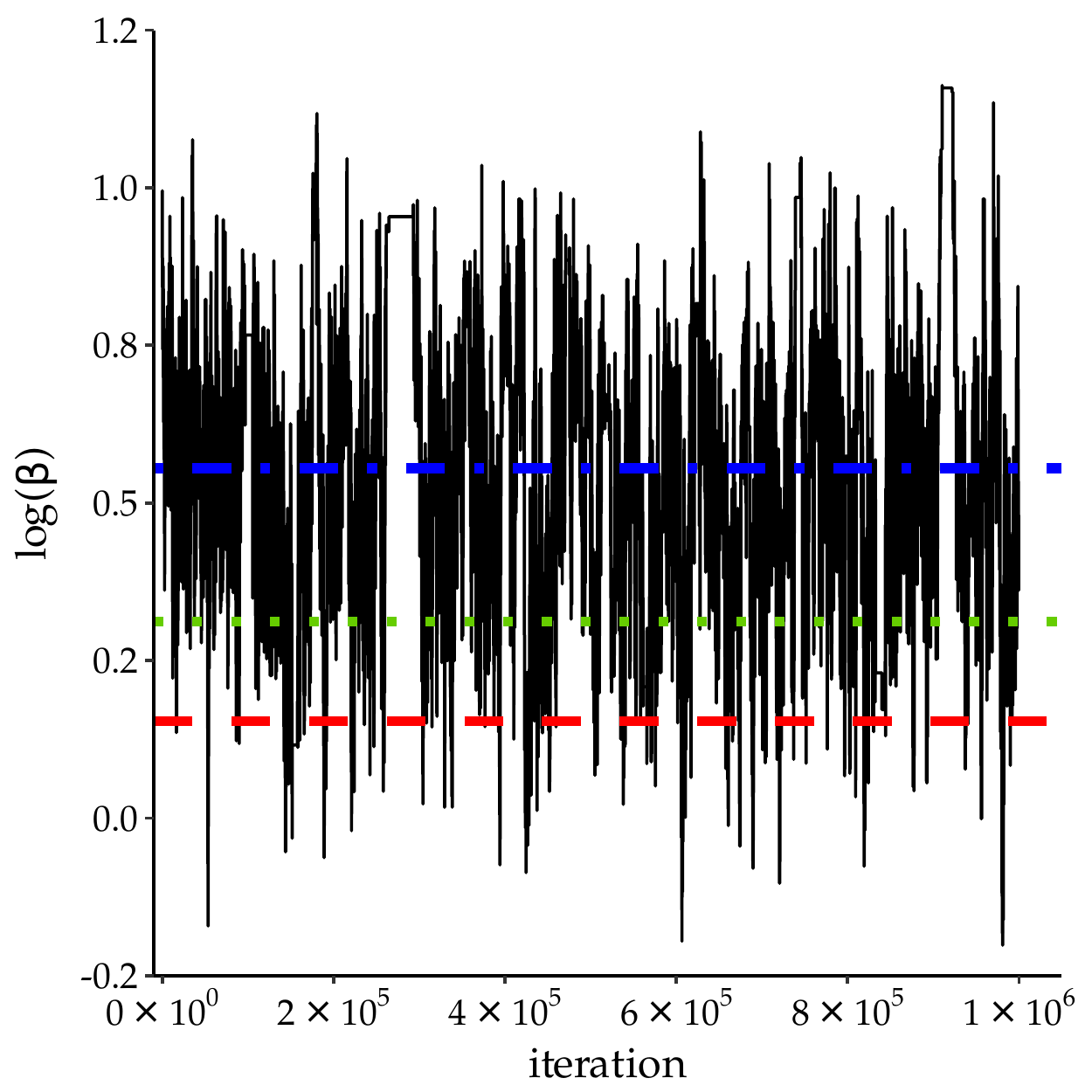}
			\caption{}
			\label{figure:04d}
		\end{subfigure}%
		\hfill
		\begin{subfigure}[t]{.49\textwidth}
			\centering
			\includegraphics[width=\textwidth]{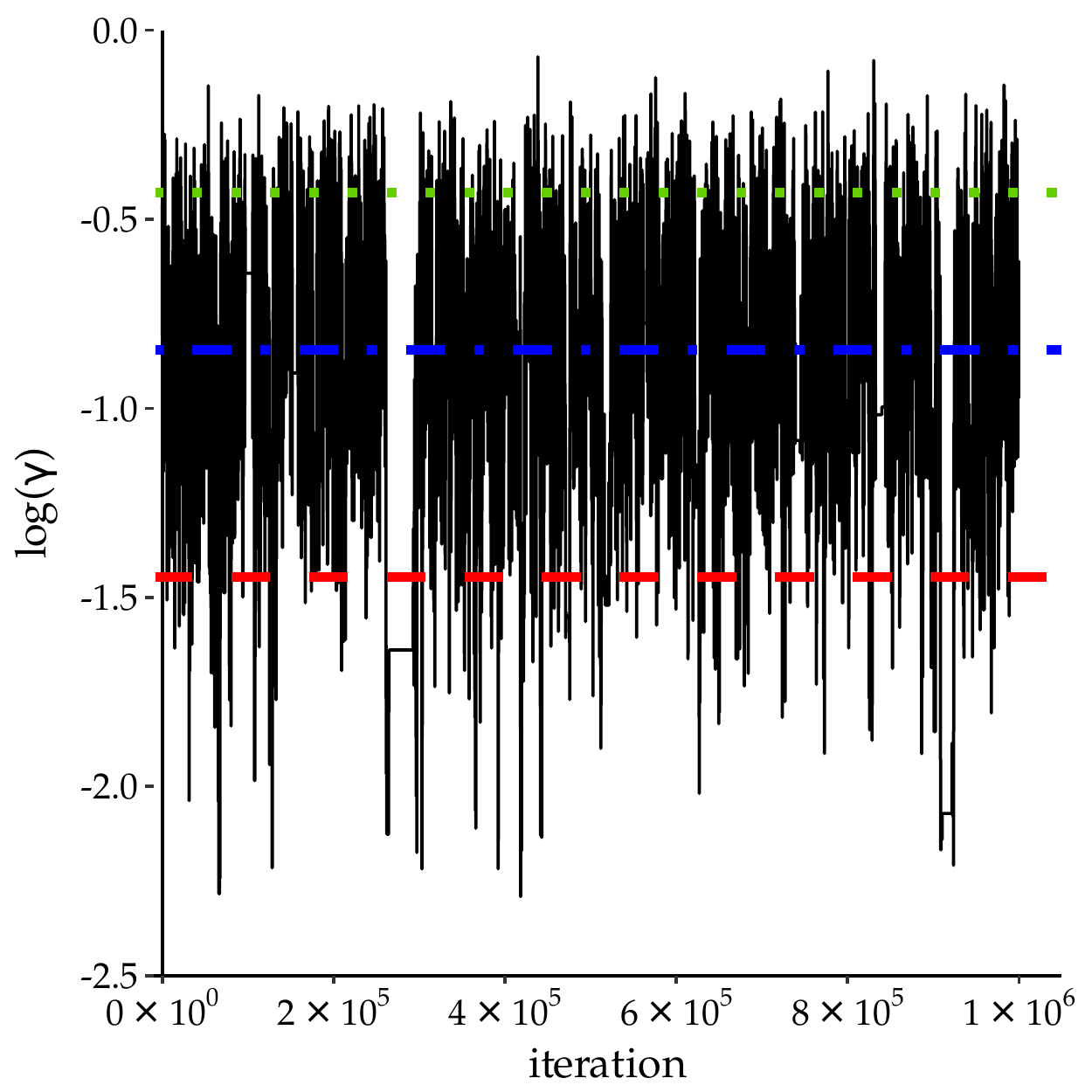}
			\caption{}
			\label{figure:04e}
		\end{subfigure}%
		\caption{\it (a) Minutiae pattern of the fingerprint from Figure \ref{figure:minutia-types}. (b) Simulated random minutiae (black) and necessary minutiae (white) using the necessary minutiae intensity computed by \eqref{eq:min.estimate} from the fingerprint image; the outer circles (dashed) have radius $R/2$. (c) Simulated minutiae with posterior probabilities in grey values from random (black) to necessary (white). The heat map gives the computed necessary minutiae intensity. (d),(e) Trace plots of the parameters $\beta$ and $\gamma$ of the Strauss process on the log-scale with true value (dashed red), MPLE (dotted green) and posterior mean (dash-dotted blue).}
		\label{figure:04}
	\end{figure}
	
	\begin{figure}[t]
		\centering
		\includegraphics[width=\textwidth]{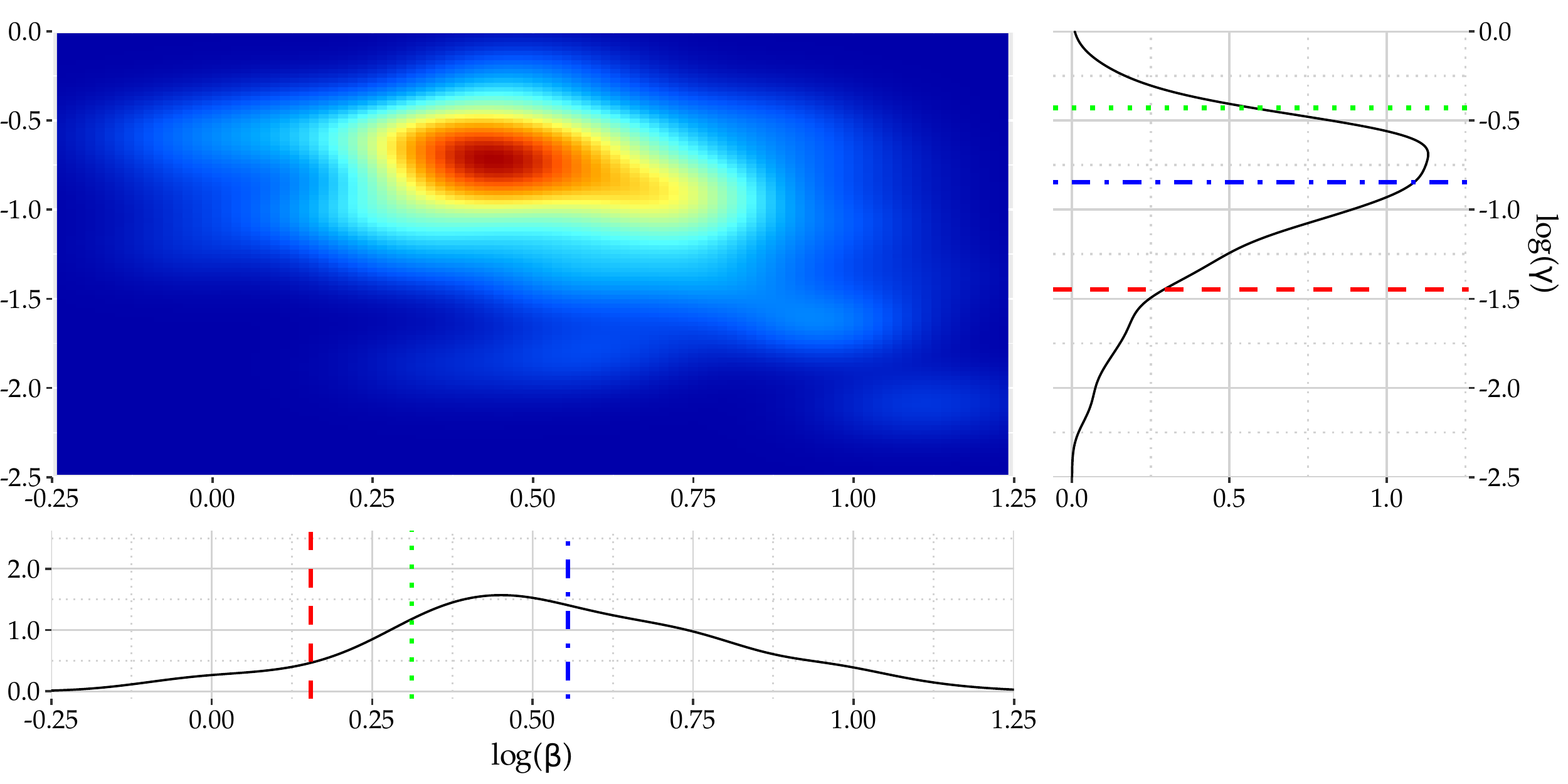}
		\caption{\it Posterior distribution of $(\beta, \gamma)$ and its marginals on the log-scale with true value (dashed red), MPLE (dotted green) and posterior mean (dash-dotted blue).}
		\label{figure:05}
	\end{figure}
	
	We determine the parameter $\hat{\boldsymbol{\theta}}$ for the auxiliary point pattern density as follows: During the burn-in phase of $10{,}000$ iterations, we compute after every $1{,}000$ iterations the MPLE $\hat{\boldsymbol{\theta}}^j = (\hat \beta^j, \hat \gamma^j)$ ($j=1,\ldots,10$) for $(\beta, \gamma)$ given the current labels $\mathbf{W}$ and use $\hat{\boldsymbol{\theta}}^j$ for the next $1{,}000$ iterations. We then use the component-wise mean $\hat{\boldsymbol{\theta}}:=\frac1{10}\sum_{j=1}^{10} \hat{\boldsymbol{\theta}}^j$ as value for $\hat{\boldsymbol{\theta}}$ for the rest of the entire run (alternatively, one could take the mean in the natural parameter space). The MPLE can be efficiently computed using the \texttt{ppm} function of the \texttt{R} package \texttt{spatstat}~\cite{baddeley_spatial_2015}.
	
	We then use the samples from the posterior of the parameter $\boldsymbol{\theta}$ as well as the labels~$\mathbf{W}$. Trace plots and the estimated posterior densities for the Strauss parameters $(\beta, \gamma)$ for the example in Figure \ref{figure:04b} are shown in Figures \ref{figure:04d}, \ref{figure:04e} and \ref{figure:05}. Overall, we observe good mixing behaviour of the Markov chain even though the parameter $\hat{\boldsymbol{\theta}}$ for the auxiliary variable method was only determined heuristically. In all 20 cases the posterior distributions concentrate around the true parameters, exemplarily shown in Figure~\ref{figure:05}. In 11 out of 20 cases $\beta$ has been overestimated as in Figure  \ref{figure:04d} and in 12 out 20 cases $\gamma$ has been overestimated as in Figure \ref{figure:04e}. This suggests that our method is not substantially biased in one or the other direction. Since the $\lambda$ updates are drawn directly from the posterior distribution, we do not show any trace plots but simply remark that estimates concentrate well in the vicinity of the true value. 
	
	The univariate marginals of the posterior distribution of the label vector $\mathbf{W}$ are depicted in panel of Figure~\ref{figure:04c} in grey values ranging from certainly random (black) to certainly necessary (white).
	While overall the minutiae separation is not too far from the truth, minutiae in regions of large necessary minutiae intensity not violating the Strauss hard core condition are more likely classified as necessary. In contrast, minutiae in regions of low intensity or those lying very close to one another are more likely classified as random.
	
	We finally apply MiSeal to real minutiae patterns of real fingerprints, including the one depicted in Figure~\ref{figure:04a}. Again we obtain good mixing behaviour for the parameter estimation. 
	Figure~\ref{figure:posterior-means} depicts the estimated posterior means for $\boldsymbol{\theta}$ for the 20 considered fingerprints which are rather spread out in the parameter space. This suggests, even though we assume that they have the same interpretation for all fingerprints, that they also depend on unobserved quantities, such as quality and resolution of the image or manner of imprinting on the acquisition medium.
	
	For each of the data fingerprints we numerically approximated the theoretical PCF of the process fitted in terms of the posterior mean based on 100 draws from the model, cf.~Figure~\ref{figure:posterior-means}. The resulting PCFs are pooled and also depicted as a red dashed curve in Figure~\ref{figure:PCF}. We observe that the PCFs of the fitted processes run mostly within the pointwise confidence band, however it seems clear that the $R$ chosen based on a pilot study is somewhat too small. The preselection of the interaction radii $h$ and $R$ leaves potential for future research. It may be desirable to adjust $R$ using more sophisticated characteristics of the individual fingerprint than just the average inter-ridge distance. At small distances, the PCFs tend to be slightly above the confidence band due to Poisson minutiae forming close pairs with other minutiae. This might not be very relevant in practice, and in any case it is hard to avoid this behaviour without making the model mathematically much more difficult. As a final remark, we note that a global way of assessing the model fit might be more desirable. One approach in this direction is given by the global envelope-based goodness-of-fit tests in \cite{myllymaki2017global} and \cite{myllymaki2019get}, which have been developed for a single observation of a point pattern. An extension of the approach provided there to replicated point patterns is beyond the scope of this paper.
	
	\begin{figure}[h!]
		\centering
		\begin{subfigure}{.49\linewidth}
			\centering
			\includegraphics[width=\linewidth]{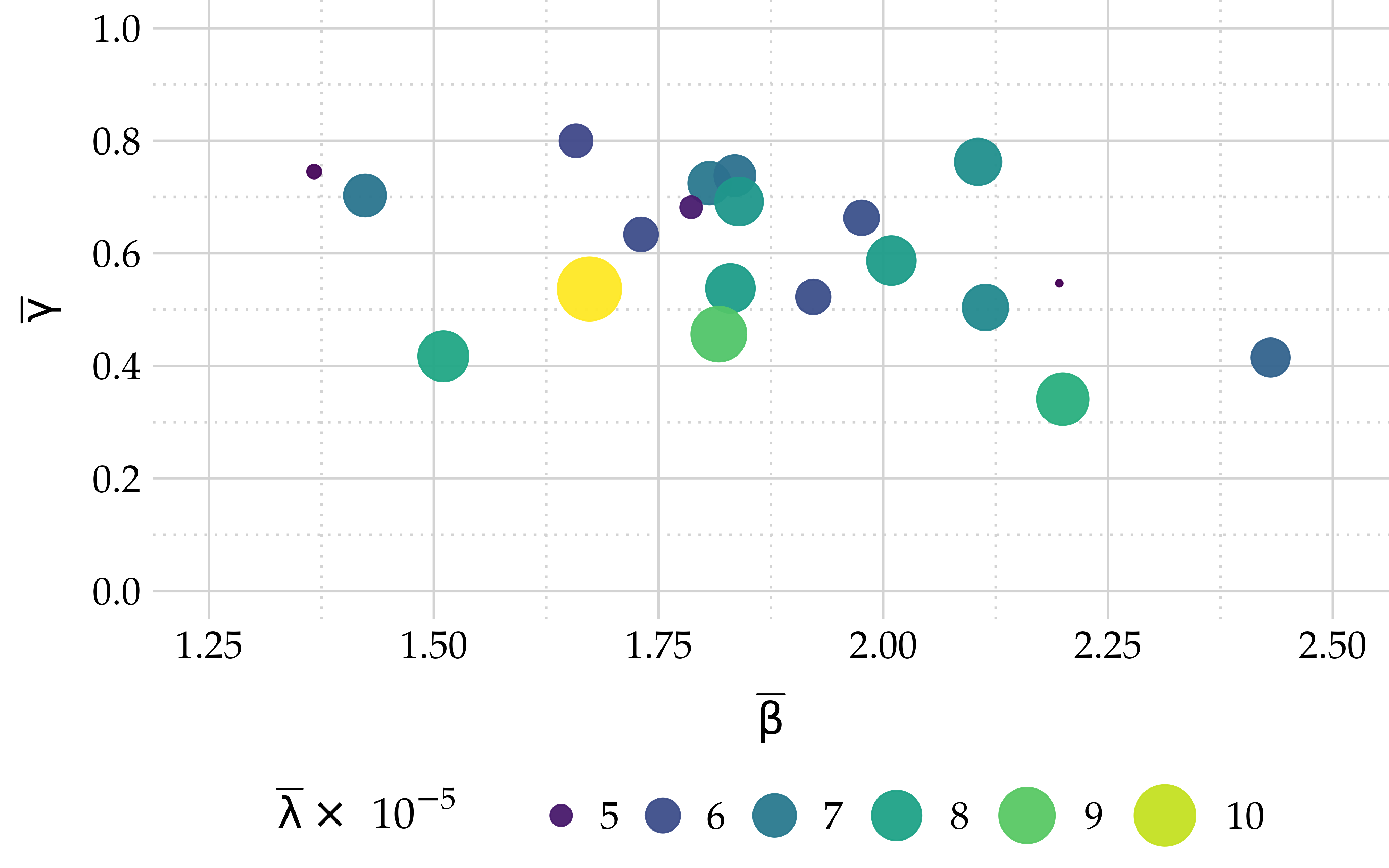}
			\caption{}
		\end{subfigure}%
		\hfill
		\begin{subfigure}{.49\linewidth}
			\centering
			\includegraphics[width=\linewidth]{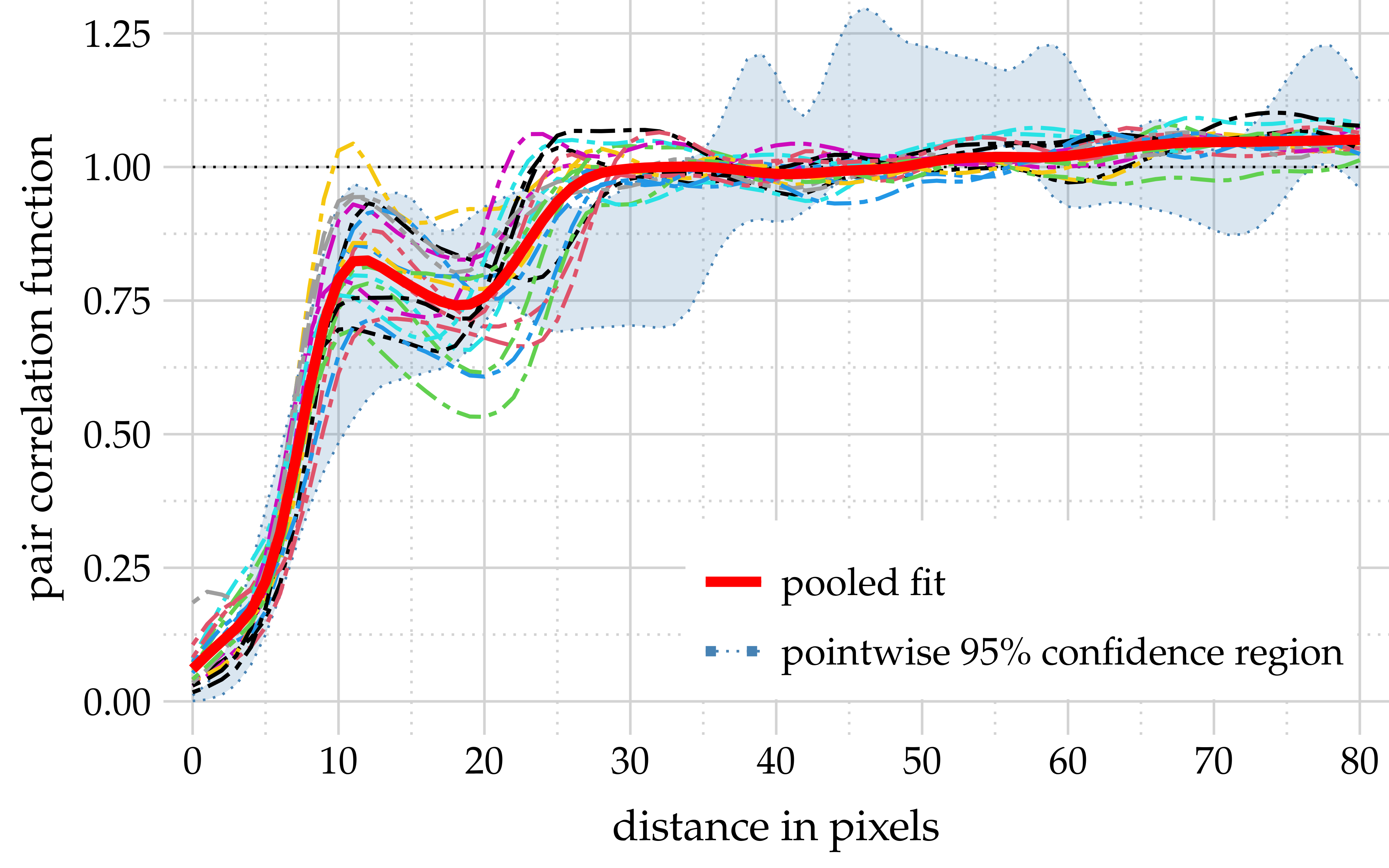}
			\caption{}
		\end{subfigure}
		\caption{\it Left: Posterior means $\bar{\boldsymbol{\theta}} = (\bar\beta, \bar\gamma,\bar\lambda)$ of $\boldsymbol{\theta} = (\beta,\gamma, \lambda)$ of the 20 fingerprints from Section~\ref{sec:exofrandom}. The value of $\bar\lambda$ is indicated as the size and colour of the bubble. Right: PCFs of the 20 fitted models (dashed) using the posterior means and their pooled PCF (red). The pointwise confidence region of the data from Figure~\ref{figure:PCF} is shaded in blue.}
		\label{figure:posterior-means}
	\end{figure}
	
	Moreover, we observe that the marginal posteriors of $\mathbf{W}$ are not independent. Considering minutiae pairs within interaction distance $R$ from each other, Fisher's exact test always rejects the hypothesis of independence at the $1\%$ level.
	
	As an example, Table~\ref{tab:1} gives the sampled posterior frequencies of the different label pairs for the two minutiae marked on Figure~\ref{figure:04a} with $\oplus$ and $\otimes$ (halfway north-west from the whorl).
	In parentheses are the expected frequencies under independence. Notably, we thinned the run by 100 (a little more than the integrated auto-correlation time), resulting  in approximately independent subsamples. We partitioned these subsamples in 100 batches and computed batch-wise \cite{matthews1975comparison} correlation coefficients to obtain a Monte Carlo estimate of the dependence between $\mathbf{W}_\oplus$ and $\mathbf{W}_\otimes$. This yields a correlation of --0.093 on average (with standard error of 0.007), suggesting negative correlation. Moreover, we compute the Kullback-Leibler divergence of the joint distribution of $(\mathbf{W}_\oplus,\mathbf{W}_\otimes)$ to the closest independent distribution resulting in a Kullback-Leibler divergence of 0.0069 (using base-2 logarithms).
	
	\begin{table}[h!]
		\centering
		\begin{tabular}{r|rr|r}
			\backslashbox{$\mathbf{W}_\oplus$}{$\mathbf{W}_\otimes$} & 0                   & 1                   & total \\ \hline
			0            & \phantom{2,3}73 \phantom{2,}(172) & \phantom{6,}636 \phantom{6,}(537) & \phantom{9,}709 \\
			1            & 2,349 (2,250) & 6,942 (7,041) & 9,291 \\ \hline
			total        & 2,422              & 7,578             & 10,000   
		\end{tabular}
		\caption{\label{tab:1}\it Contingency table of two selected components of $\mathbf{W}$  (labelled $\oplus$ and $\otimes$ in Figure~\ref{figure:04a}, halfway north-west from the whorl) whose minutiae lie within interaction distance $R$, with frequencies under hypothesis of independence in parentheses.}
	\end{table}
	
	Judging from this analysis, it seems important to consider the whole distribution of $\mathbf{W}$ provided by Algorithm~\ref{algo:MCMC}, rather than only the marginals obtained by the method in \cite{rajala_variational_2016}.
	
	\section{Random Minutiae Are Characteristic}
	\label{sec:proof}
	
	Here we consider the two fingerprints from \cite{newman_finger_1930}, already shown in Figure~\ref{figure:twin-fingerprints}. At first glance, they appear very similar based on their OFs, but actually stem from two different persons.
	
	As before, we enhance these images and extract the minutiae manually. Then, we approximate the posterior distributions $\pi_1$, $\pi_2$ of the label vectors $\mathbf{W}$ with our MiSeal (Section~\ref{sec:mcmc}). Their marginal probabilities are depicted in the right bottom of Figure~\ref{figure:Intensity-Twins}. In particular the north-east part of the right print in Figure~\ref{figure:Intensity-Twins} contains candidates for random minutiae (from grey to black), that are not found on the left print. Notably, such candidates tend to cluster which indicates their high correlation. In a realization, however, within such a cluster only as many minutiae will be random as the sum of marginal probabilities indicates, so that most clusters disappear (compare also Figure \ref{figure:04b} with Figure \ref{figure:04c}), making the random minutiae pattern Poissonian.
	
	\begin{figure}[b!]
		\centering
		\begin{subfigure}{.49\textwidth}
			\centering
			\includegraphics[width = \linewidth]{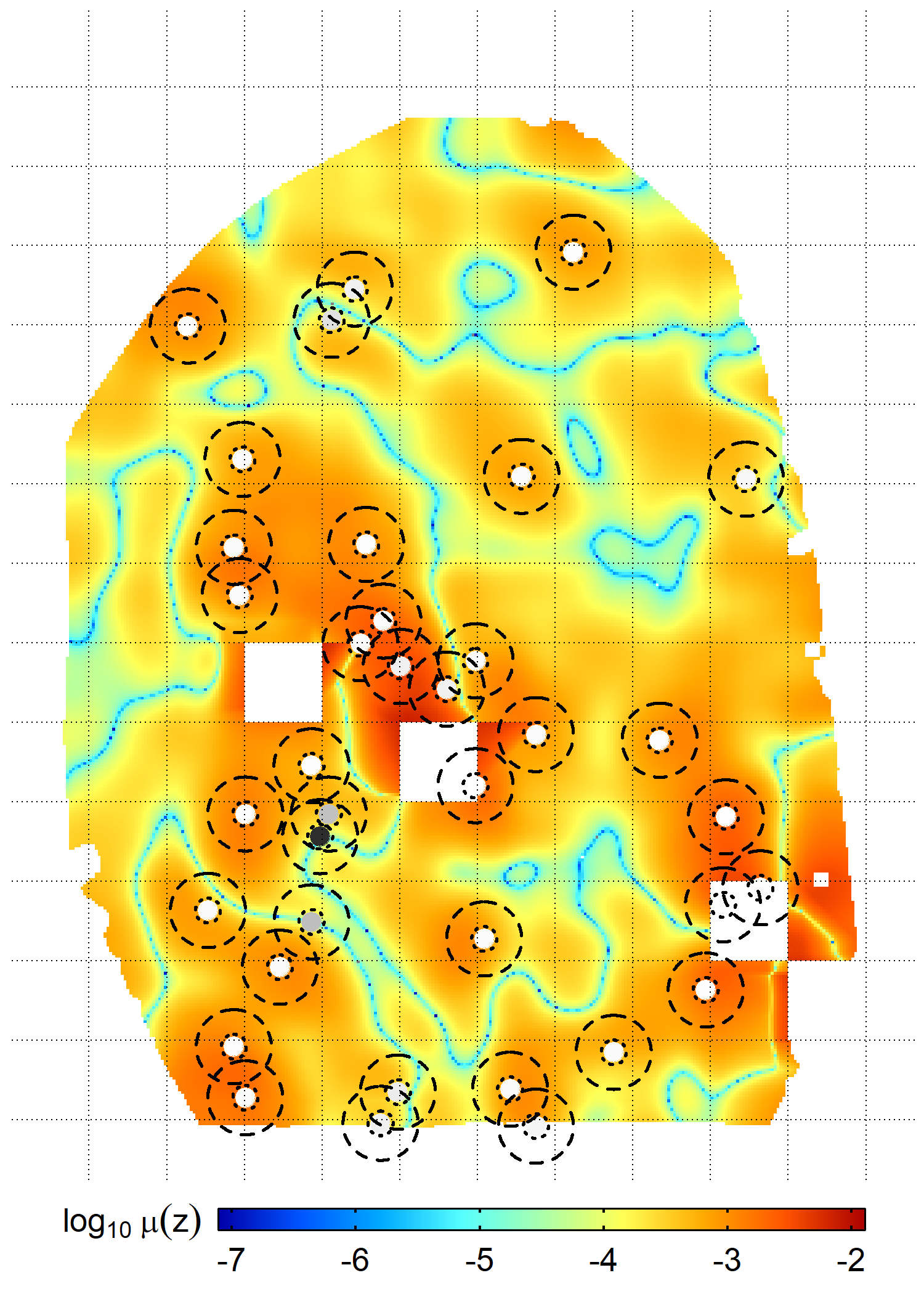}
		\end{subfigure}%
		\hfill
		\begin{subfigure}{.49\textwidth}
			\includegraphics[width = \linewidth]{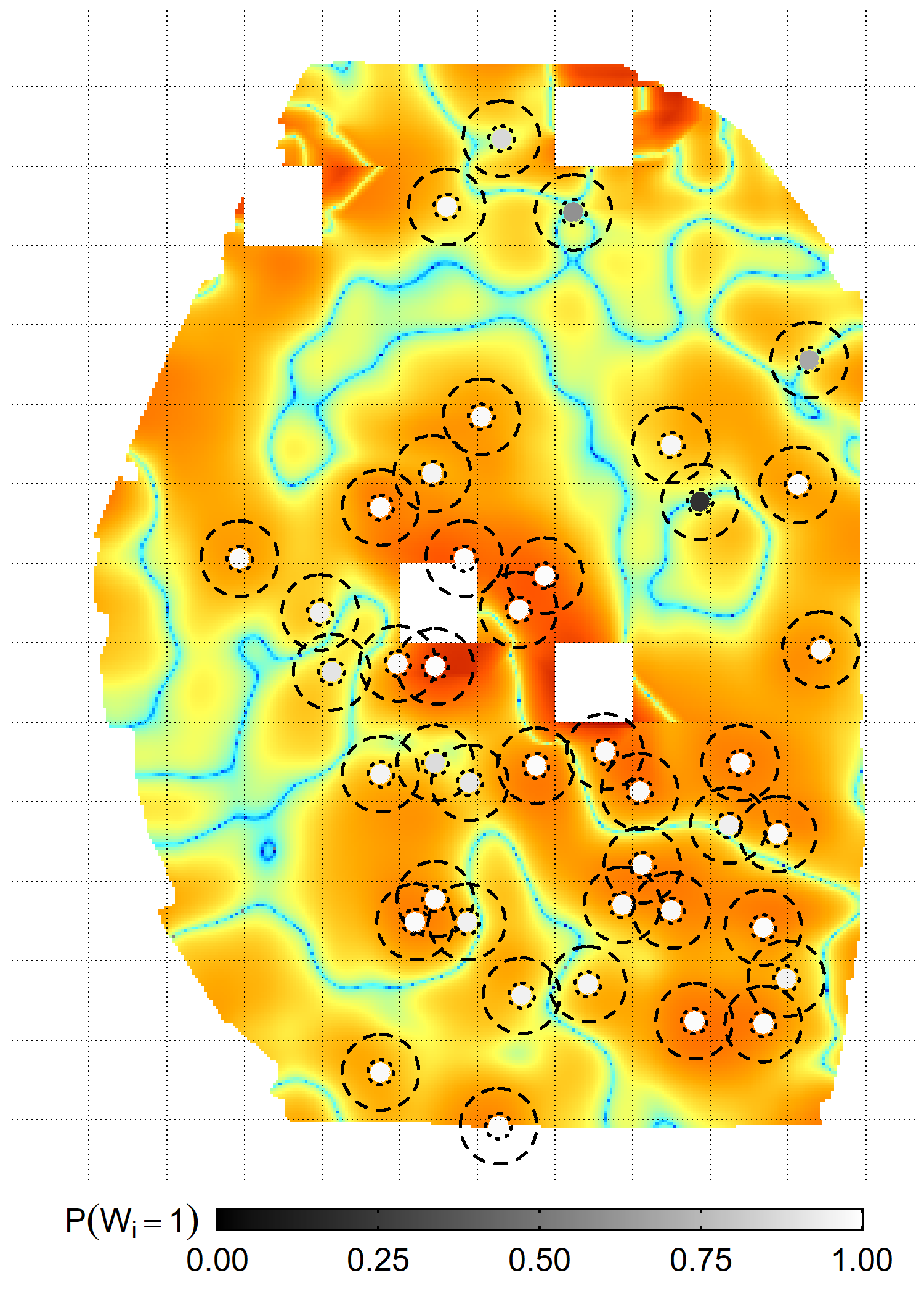}
		\end{subfigure}
		\caption{\textit{The necessary minutiae intensities of the two twin 
				fingerprints from Figure~\ref{figure:twin-fingerprints} as heat maps on 
				the log-scale from blue (low) to red (high). Marginal posterior 
				probabilities $\pi_i$ are indicated on a greyscale from black (probably 
				random) to white (probably necessary). The dashed circles have radius 
				$R/2$, i.e.\ intersecting circles indicate minutiae within interaction 
				distance $R$. Note that the intensity is computed patch-wise (lattice 
				indicated as dotted lines) and patches containing singularities were 
				excluded.}}
		\label{figure:Intensity-Twins}
	\end{figure}
	
	For comparison of two minutiae patterns we use the Minutiae Cylinder Code (MCC) matching algorithm obtained from \cite{cappelli_mcc_2010}, which is publicly available. The MCC compares two given minutiae templates $\zeta^{(1)},\zeta^{(2)}$ exploiting local information, i.e. spatial and directional similarity of minutiae and their neighbourhood (cylinders) and combines the most similar cylinders to a global score $S(\zeta^{(1)}, \zeta^{(2)}) \in [0,1]$ where 1 means very similar and 0 means very different. In order to assess \emph{characteristicness} of random minutiae, we investigate whether deleting random minutiae leads to more similar fingerprints than deleting the same number of arbitrary minutiae. For this, we repeat the following procedure $1{,}000$ times:
	\begin{enumerate}[label = (\arabic*)]
		\item Draw a sample $\mathbf{W}^{(i)} \sim \pi_i$ from the posterior of the labels and let $r^{(i)}$ be the number of random minutiae in the minutiae template $\zeta^{(i)}$, $i=1,2$.
		\item Delete from the minutiae template $\zeta^{(i)}$ the minutiae labelled as random under $\mathbf{W}^{(i)}$ to obtain a new template $\zeta^{(i,n)}$ containing only the necessary minutiae, $i=1,2$.
		\item Draw uniformly at random $r^{(i)}$ minutiae from $\zeta^{(i)}$ and delete them from $\zeta^{(i)}$ to obtain a new template $\zeta^{(i,r)}$ having the same number of minutiae as $\zeta^{(i,n)}$, $i=1,2$.
		\item Compute the matching scores $S^{(n)} :=S(\zeta^{(1,n)}, \zeta^{(2,n)})$ and $S^{(r)} := S(\zeta^{(1,r)}, \zeta^{(2,r)})$ using the MCC. 	
	\end{enumerate}
	
	We then compute the differences between these 1{,}000 pairs of matching scores. Note that local clusters of minutiae, which make a major contribution to the MCC matching score, are often dissolved by our deletion scheme, leading to scores on a very small scale. We therefore consider the relative score differences in Figure~\ref{figure:score-differences}. We obtain a Monte Carlo estimate of 93.6\% (with standard error of 0.8\%) that  the matching score after deletion of random minutiae is larger than the score after deleting the same number of minutiae at random. 
	Similarly, the Monte Carlo estimate for the relative difference of scores yields  a 23.7\% improvement (with standard error of 0.6\%).
	
	\begin{figure}[h!]
		\centering
		\includegraphics[width=\linewidth]{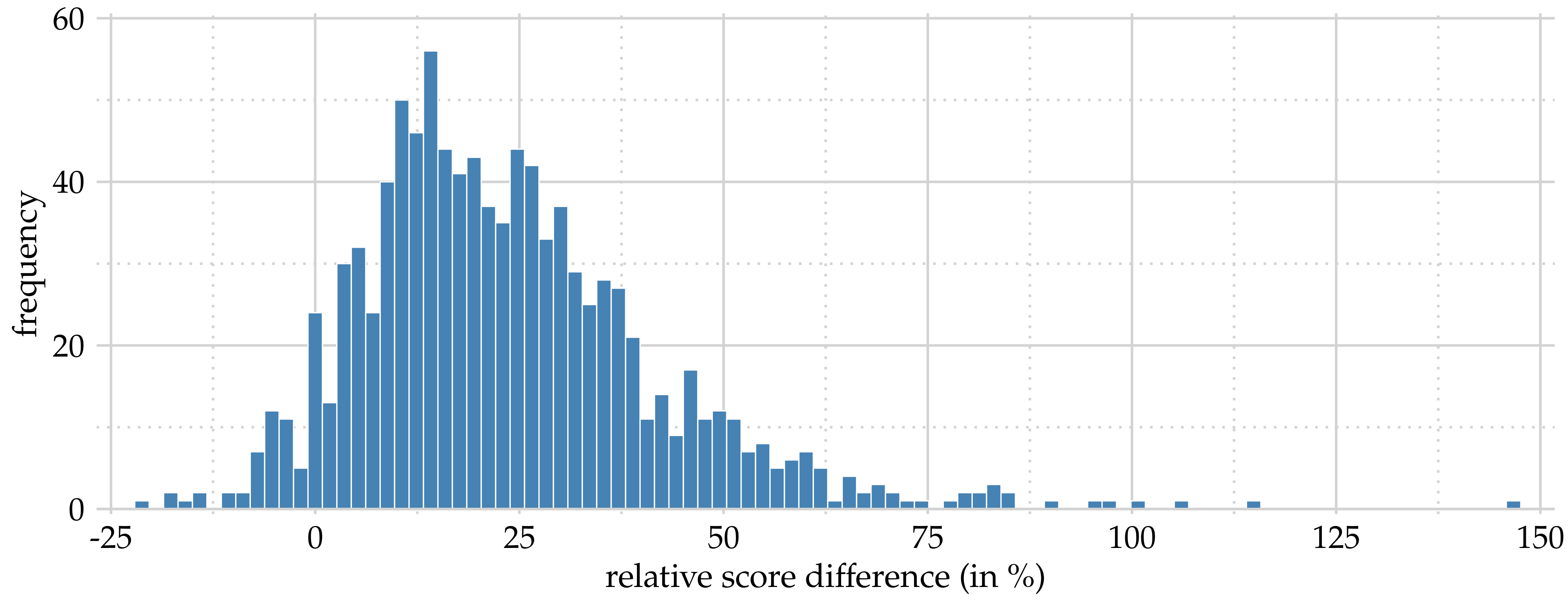}
		\caption{\textit{Histogram of relative score differences $\left( S^{(n)}-S^{(r)} \right) \slash S^{(r)}$ (in \%). Overall 93.6\% of these distances are positive and the average difference is 23.7\%.}}
		\label{figure:score-differences}
	\end{figure}
	
	Hence, we may conclude that the two different fingerprints become more similar to one another after deleting the random minutiae in comparison to just randomly deleting minutiae. This preliminary proof of concept suggests that for fingerprints with similar OFs the information encoded in random minutiae is \emph{characteristic} to distinguish them from one another.
	
	\section{Discussion}
	\label{sec:discussion}
	
	We have introduced a model which provides a formula for predicting locally the number of necessary minutiae determined by OF and RF of a fingerprint. In a statistical analysis, based on 20 high quality images, we have found that fingerprints feature additional \emph{random} minutiae. By considering the pair correlation function for the same data, we have concluded that it is reasonable to model the necessary minutiae by a Strauss process with hard core, while the additional random minutiae can be modelled by a homogeneous Poisson point process.
	
	For the independent superposition of the two processes, we can apply an MCMC algorithm for exploring the distribution of necessary and random minutiae of a given fingerprint as well as, simultaneously, the model parameters. The proposed MiSeal  (Section~\ref{sec:mcmc}) is based on the work of \cite{redenbach_classification_2015} but provides significant improvements in terms of mixing times and does not need assumptions on the independence of the components of the label vector as in \cite{rajala_variational_2016}. A crucial ingredient for good mixing seems to be a good choice of the marginal distribution for the auxiliary point pattern. The associated parameter $\hat{\boldsymbol{\theta}}$ also has to be chosen appropriately which we achieve by repeated estimation during the burn-in phase. As a future improvement, updating of $\hat{\boldsymbol{\theta}}$ can also be performed during the entire run if we let the adjustments diminish and adapt the Hastings ratio for the $\mathbf{W}$-update appropriately. 
	
	It turned out that, for two similar yet different fingerprints, excluding random minutiae yields a highly significant improvement of the similarity score as compared to excluding arbitrary minutiae. This suggests that the random minutiae carry characteristic information of fingerprint individuality going beyond OFs and RFs, which is why we refer to them as \emph{characteristic minutiae}. 
	
	The extent to which this information can effectively be used for discriminating different fingerprints with similar OFs is the subject of current and future research. One important ingredient will be sufficiently robust minutiae extraction.
	
	Additionally, various parameters of our MiSeal, for instance the smoothing of the necessary minutiae intensity, can be more finely tuned based on larger data sets. Eventually, we expect that including the degree of characteristicness of minutiae will improve error rates of minutiae matching algorithms.
	
	\section*{Acknowledgements}
	
	Johannes Wieditz gratefully acknowledges support by the DFG Research Training Group 2088 ``Discovering structure in complex data: Statistics meets Optimization and Inverse Problems''. The first and the last author further gratefully acknowledge support by the Niedersachsen Vorab of the Volkswagen foundation and the Felix-Bernstein-Institute of Mathematical Statistics in the Biosciences. Yvo Pokern also wishes to thank the Royal Society for International Exchanges grant IE150666. We thank Claudia Redenbach for helpful discussions and for providing the code from \cite{redenbach_classification_2015}, \cite{rajala_variational_2016}. The authors are also grateful to Corvin Grigutsch for writing large parts of the software. In conclusion, we thank the three anonymous referees for their very helpful comments improving the paper. 
	
	{
		\bibliographystyle{rss}
		\bibliography{references}
	}
	
\end{document}